\documentclass[aps,pra,twocolumn,amsmath,amssymb,showpacs,superscriptaddress,notitlepage,longbibliography]{revtex4-1}
\usepackage[colorlinks=true,linkcolor=blue,anchorcolor=red,citecolor=blue, urlcolor=blue]{hyperref}
\usepackage{bm}
\usepackage{graphicx}
\usepackage{color}
\usepackage{soul}
\usepackage{mathrsfs}
\usepackage{bbm}
\usepackage{braket}
\usepackage{array}
\usepackage{tikz-cd}

\newcommand{\img}{\text{im}}

\newcommand{\wt}{\text{wt}}
\usepackage{amsthm}
\newtheorem{thm}{Theorem}
\newtheorem{dfn}[thm]{Definition}
\newtheorem{lemma}[thm]{Lemma}
\newtheorem{conjecture}[thm]{Conjecture}
\newtheorem{proposition}[thm]{Proposition}

\begin{document}
\title{Improved energy barrier in higher-dimensional hypergraph product codes}

\author{Guangqi Zhao}
\affiliation{Centre for Engineered Quantum Systems, School of Physics,
University of Sydney, Sydney, NSW 2006, Australia}
\email{guangqi\_zhao@outlook.com}

\begin{abstract}
Single-shot error correction outperforms conventional approaches by requiring only one round of stabilizer measurements for decoding, even in the presence of measurement errors. This capability relates to the soundness and confinement property of codes, which provides an energy barrier lower bound. Earlier research established a confinement property for 3D hypergraph product codes [Quintavalle et al. 2021 PRX Quantum], yielding an energy barrier lower bound for these codes. In this work, by analyzing the structure of logical operators, we show an improved energy barrier lower bound for higher-dimensional hypergraph product (HHGP) codes with low-density parity check (LDPC) property. Our bound exceeds results derived from soundness and confinement, and unlike standard hypergraph product codes, these higher dimensional variants can possess macroscopic energy barriers even when the underlying classical codes lack this property. Specifically, our results show that the energy barrier of LDPC HHGP codes is lower bounded by the distance of the underlying classical codes. This bound is tight if the underlying classical codes exhibit system size-dependent distances but constant energy barriers, like 3D and 4D toric codes.
\end{abstract}

\maketitle

\section{Introduction}

Quantum error correction is crucial for fault tolerant quantum computing, yet it introduces significant overhead. Within the stabilizer codes \cite{gottesman1997stabilizer} implemented through the quantum circuit model~\cite{nielsen2001quantum}, a significant portion of overhead stems from the stabilizer measurements required to detect and identify errors. Specifically, along the quantum circuit, stabilizer measurements and decoding must be performed at regular intervals to maintain the integrity of the code state. Moreover, to overcome measurement errors, conventional protocols require multiple rounds of stabilizer measurements for each decoding process, with the number of rounds typically scaling with the code distance~\cite{shor1996fault,divincenzo1996fault}.

In response to this overhead, single-shot error correction emerges as a promising protocol by requiring only one round of stabilizer measurements per decoding process~\cite{bombin2015single}, even in the presence of measurement errors. Although this approach still leaves some residual error after each decoding, the residual error remains bounded. Such controlled error propagation is sufficient for fault-tolerant quantum computation.

Single-shot error correction requires quantum codes with specialized structures. Bomb\'{i}n first demonstrated that codes with confinement property enable single-shot correction~\cite{bombin2015single}. Campbell later proposed that good soundness is also sufficient for single-shot error correction~\cite{campbell2019theory}. He also showed that single-shot error correction can be realized for any code by measuring a carefully chosen set of stabilizers. Further advancing the theory, Quintavalle et al. provided a more general definition of confinement, and proved that confinement is sufficient for single-shot decoding in the case of adversarial errors~\cite{Quintavalle2021single}.

So far, codes shown to support single-shot error correction include the 3D gauge color code~\cite{bombin2015single,Ben2016nc}, 3D hypergraph product codes (only for either X or Z noise, depending on code structure)\cite{kubica2019cellular,vasmer2021cellular,Quintavalle2021single,Higgott2023improved}, 4D hypergraph product codes\cite{niko2017local,campbell2019theory,zeng2019higher}, 4D hyperbolic codes~\cite{Breuckmann_2022}, quantum expander codes~\cite{fawzi2020constant}, and quantum Tanner codes~\cite{gu2024single}.

Among these codes, higher-dimensional hypergraph product (HHGP) codes~\cite{zeng2019higher} are an important family of quantum codes that extend the original hypergraph product construction~\cite{2014-Tillich} into higher dimensions. While the standard hypergraph product combines only two classical codes to form a quantum code, this higher-dimensional approach incorporates additional classical codes. These additional structure provide meta-checks for X and Z stabilizers, effectively validating the stabilizers themselves, thereby providing an intuitive foundation for the single-shot error correction capability these codes exhibit.

Single-shot error correction shares fundamental connections with self-correcting memory~~\cite{ben2016rmp,terhal2015quantum}, which requires the system to preferentially stay within the code space at finite temperature. A thorough analysis of self-correction requires evaluation of free energy dynamics. While entropic effects play a crucial role in determining the system's stability~\cite{plischke1994equilibrium,bravyi2013,Bravyi.2011}, a macroscopic energy barrier constitutes a necessary condition for self-correction. The term ``macroscopic'' means that energy barrier scales with system size.

Previously, the author established a lower bound for the energy barrier of hypergraph product codes with low-density parity check (LDPC) property~\cite{zhao2024on}. In that case, the quantum energy barrier is lower bounded by the energy barriers of the underlying classical codes. Thus, constructing quantum codes with macroscopic energy barriers necessitates classical codes with the same property. However, finding the classical codes with macroscopic energy barriers remains challenging, raising the question of how to construct quantum codes that exhibit macroscopic energy barriers when the classical codes do not.

The HHGP codes indeed meet this criterion. In this work, we will show that the energy barrier of LDPC HHGP codes is lower bounded not only by the energy barriers of the underlying classical codes (as in standard hypergraph products), but also by the distances of these underlying classical codes. Thus to construct quantum codes with macroscopic energy barrier, one can just use classical codes that have large distances.

Notably, previous research has demonstrated that HHGP codes exhibit good soundness and confinement, making them single-shot correctable against adversarial errors~\cite{campbell2019theory,Quintavalle2021single}. These soundness and confinement properties induce a lower bound on the energy barrier. However, we demonstrate that this bound is not optimal. Our approach, based on analyzing logical operator structures, provides a tighter characterization of the energy landscape for HHGP codes.

To guide the reader, our paper is structured as follows. In Sec.~\ref{sec:preliminaries}, we review key concepts including soundness, confinement, energy barrier and HHGP codes. In Sec.~\ref{sec:confinement_and_energy_barrier}, we formally demonstrate how confinement induces a lower bound on energy barriers and use the confinement results from Ref.~\cite{Quintavalle2021single} to establish an energy barrier lower bound for 3D hypergraph product codes. Then, in Sec.~\ref{sec:tensor_product_code_and_its_energy_barrier}, we provide an energy barrier lower bound for tensor product codes, as formally stated in Lemma~\ref{lemma:energy_barrier_of_tensor_product_codes}. This lower bound serves as a key insight for proving the energy barrier lower bound of LDPC HHGP codes. Finally, by analyzing the structure of logical operators of HHGP codes and applying Lemma~\ref{lemma:energy_barrier_of_tensor_product_codes}, we derive energy barrier lower bounds for 3D and 4D hypergraph product codes in Sec.~\ref{sec:energy_barrier_of_3d_hgp} and Sec.~\ref{sec:energy_barrier_of_4d_hgp}, respectively. These bounds, presented in Theorem~\ref{thm:energy_barrier_3d_hgp} and Theorem~\ref{thm:energy_barrier_4d_hgp}, relate directly to the distances of the underlying classical codes.

\section{Preliminaries}
\label{sec:preliminaries}

\subsection{Quantum LDPC codes}

Stabilizer codes, conceptualized by Daniel Gottesman in 1996~\cite{gottesman1997stabilizer}, encode quantum information in a protected subspace of a larger Hilbert space. Specifically, these codes store quantum states within the $+1$ co-eigenspace of a set of $m$ commuting Pauli operators, $\mathcal{S} = \{S_1, S_2, \ldots, S_m\}$, defined on an $n$-qubit Hilbert space $\mathcal{H}_2^{\otimes n}$. Each operator $S_i \in \{I, X, Y, Z\}^{\otimes n}$ consists of a tensor product of the Pauli operators, with the requirement that all these operators commute with each other. $\mathcal{S}$ excludes the operator $-I$ to ensure a non-trivial code space. Logical operators are defined as the elements of $\mathcal{C}(S) \backslash S$.

Calderbank-Shor-Steane (CSS) codes \cite{calderbank1996good,steane1996error} are an important family of stabilizer codes with special structural properties. Their stabilizer generators are exclusively either X-type or Z-type Pauli operators. This separation allows all X-type stabilizers to be represented by a parity check matrix $H_X$, while all Z-type stabilizers are represented by a parity check matrix $H_Z$. Consequently, the parity check matrix for the quantum CSS code takes the form
\begin{eqnarray}
H = \left(\begin{array}{cc}
0 & H_Z \\
H_X & 0
\end{array}\right),
\end{eqnarray}
with the condition $H_XH_Z^T = 0$ ensuring that all stabilizers commute with each other.

Quantum LDPC codes are stabilizer codes with sparse parity check matrices (see \cite{zunaira2015,Breuckmann2021ldpc} for good reviews). The sparsity is defined by two parameters: $w_c$, the maximum weight amongst all stabilizer generators, and $w_q$, the maximum number of stabilizer generators associated with any single qubit. Formally, a code is classified as LDPC if both $w_c$ and $w_q$ remain constant as the code size increases, specifically $w_c, w_q=O(1)$.

Quantum LDPC codes are promising because they often provide lower overhead and can be designed for fault-tolerance~\cite{gottesman2014faulttolerant}. Recent research has led to several families of quantum LDPC codes with enhanced parameters~\cite{freedman2002z2,evra2022decodable,kaufman2021new,hastings2021fiber,panteleev2021quantum,Breuckmann_2021,bacon2017sparse}, with some efforts even achieving good quantum LDPC codes~\cite{Panteleev.2021,leverrier2022quantum,dinur2022good}.

Various methods exist for constructing quantum LDPC codes. From a mathematical perspective, the hypergraph product construction~\cite{2014-Tillich,zeng2019higher} represents a well-established approach, providing a systematic way to construct quantum codes from any classical codes. Alternative methods such as two-block group-algebra (2BGA) codes~\cite{wang2023abelian,lin2024twoblock} offer different advantages,  like suitability for two-dimensional layouts~\cite{bravyi2024bbcode}, making them useful for practical quantum computing architectures.

\subsection{Energy barrier of codes}

Following Ref.~\cite{Bravyi2009nogo},  we present a formal definition of the energy barrier for stabilizer codes. Let $\mathcal{C}$ be the code subspace of a stabilizer code, defined in terms of the stabilizer group $S$. The code subspace can be viewed as the ground state subspace of a Hamiltonian of the form $\hat{H} = \sum_{i=1}^{m}(I-s_i)/2$, where $\{s_1, \ldots, s_m\} \subset S$ is a set of stabilizer generators.

For an operator $P$ in the Pauli group $\mathcal{P}$, the energy of the state $P|\psi\rangle$ is given by $\bra{\psi}P^\dag \hat{H} P\ket{\psi} = \epsilon(P)$. Here $\ket{\psi}$ is any ground state with $\bra{\psi}\hat{H}\ket{\psi} = 0$, and $\epsilon(P)$ is the energy cost of $P$. This is the number of $s_i$s that anti-commute with $P$, i.e., $\epsilon(P) = | \{i: s_iP = -Ps_i\}|$. Equivalently, one may define the energy barrier in terms of the parity check matrix $H$ of the stabilizer code. Let $v(P)$ be the binary (bit-string) representation of a Pauli $P$, then
\begin{equation}
    \epsilon(P) = \text{wt}(Hv(P)).
\end{equation}

A sequence $P_0, P_1, \ldots, P_t$ from the Pauli group $\mathcal{P}$ forms a \emph{path} from $P_0$ to $P_t$ if for each index $i$, the operators $P_i$ and $P_{i+1}$ differ at no more than one qubit, we call it local error condition. The notation $w(P_0, P_t)$ represents the collection of all such paths from $P_0$ to $P_t$. For $ r \in w(P_0, P_t)$, $\epsilon_{\max }(r)$ denotes the highest energy along path $r$, i.e., $\epsilon_{\max}(r) = \max_{P_i \in r}\epsilon(P_i)$, as the energy barrier of $E$ along the path $r$.

The minimum energy associated with a Pauli $P$ is the smallest value of $\epsilon_{\text{max}}$ across all possible paths from $I$ to $P$. This is the energy barrier of $P$, denoted as $\Delta(P)$,
\begin{eqnarray}
  \Delta(P) = \min_{r\in w(I,P)}\epsilon_{\max}(r).
\end{eqnarray}

The energy barrier of the quantum code is the minimum energy barrier over the set of nontrivial logical operators.
\begin{dfn}
   Let $S$ be a stabilizer group and $L(S)$ be the set of nontrivial logical operators. The energy barrier is  
   \begin{eqnarray}
      \Delta(H) := \min \{\Delta(\ell): \ell \in L(S)\}.
   \end{eqnarray}
   \label{dfn:quantum_energy_barrier}
\end{dfn}
This is the smallest energy the environment has to overcome to enact a logical operation on the encoded qubit. We can similarly define the energy barrier of classical codes by only considering the path formed by Pauli-Xs. We shall denote this energy barrier also as $\Delta(H)$, where $H$ in this case is the parity check matrix of the classical code.

There are two main obstacles for computing the energy barrier of quantum codes. First, numerous paths exist for implementing any logical operator, creating a vast solution space to explore. Second, in quantum codes, any stabilizer can be multiplied with a logical operator to produce an equivalent logical operator, while for those equivalent logical operators, the energy barriers may differ. That is to say, for a logical operator $L$ and stabilizer $S$, usually  $\Delta(L) \neq \Delta(LS)$. This complicates the analysis as one must consider the full equivalence class of logical operators.

Quantum LDPC codes offer significant advantages for energy barrier computation by reducing analytical complexity. As shown in Ref.~\cite{zhao2024on}, in LDPC codes, two logical operators that are equivalent under stabilizers have identical energy barriers, provided at least one of their energy barriers equals or exceeds $w_c w_q$, the product of sparsity parameters. With $w_c w_q=O(1)$, this property yields a powerful simplification: if any non-trivial logical operator has an energy barrier of $\Omega(1)$, all equivalent logical operators share the same energy barrier.

Therefore, to determine the energy barriers of quantum LDPC codes, one needs only analyze a fixed complete set of logical operators. Once the energy barrier is established for this complete set, the energy barrier of the code follows directly from these results.

\subsection{Soundness}

Soundness originated in the study of locally testable codes (LTCs)~\cite{blum1990self,arora1994probabilistic,rubinfeld1996robust,friedl1995some}. LTCs are a class of error-correcting codes that allow efficient verification of whether a string is a valid codeword or is far from the code space. The soundness of a code, denoted by $R(\delta)$, measures the likelihood that a randomly selected local constraint is violated by a string that is at Hamming distance at least $\delta n$ from the code space, where $n$ is the codeword length. The concept of locally testability has deeply influenced fields like PCP, combinatorial optimization, property testing, program verification, and cryptography.

Aharonov and Eldar introduced quantum locally testable codes (QLTCs)~\cite{aharonov2015quantum}, with interest surging after Eldar and Harrow showed~\cite{eldar2017local} that QLTCs with constant soundness, locality, and relative distance could construct Hamiltonians lacking low-energy trivial states—addressing the NLTS conjecture~\cite{freedman2014quantum}. QLTCs are closely related to the quantum PCP conjecture \cite{aharonov2013guest}, a fundamental problem in quantum complexity theory related to the quantum analog of the classical PCP theorem.

Roughly speaking, the soundness property ensures that high-weight errors produce high-weight syndromes. For stabilizer codes, soundness can be characterized as follows. Consider a code $C$ defined on $n$ qubits, determined by a generating set $\mathcal{G}$ where each generator has support over $w$ qubits (representing the stabilizer size). For any error $E$ acting on the code with a syndrome weight of at least $\delta n$, a randomly chosen generator $g \in \mathcal{G}$ fails to commute with $E$ with probability at least $R(\delta)$. This property ensures that errors can be detected probabilistically by examining only a limited number of stabilizers. Additionally, this soundness property induces a macroscopic energy barrier~\cite{rakovszky2023physics}.

To analyze single-shot property of quantum codes, Campbell provided the following alternative definition of soundness~\cite{campbell2019theory}.

\begin{dfn}[Soundness, Definition 3 in~\cite{campbell2019theory}]
    Let $t$ be an integer and $f: \mathbb{Z} \rightarrow \mathbb{R}$ be a function termed the soundness function, where $f(0)=$ 0. A set of Pauli stabilizers $\mathcal{M}$ is called $(t, f)$-sound if, for every Pauli error $E$ with syndrome weight $|\sigma(E)|=x< t$, there exists an $E^{\star}$ such that $\sigma\left(E^{\star}\right)=\sigma(E)$ and the weight wt $\left(E^{\star}\right) \leqslant f(x)$.
    \label{soundness:ss}
\end{dfn}

In essence, this definition requires that low-weight syndromes are caused by low-weight errors. This formulation represents the contrapositive of definition of locally testable codes, with the key distinction being its focus on scenarios where $|\sigma(E)|=x< t$. Under this definition, good soundness is characterized by function $f(x)$, which is a monotonically increasing polynomial function of $x$ that remains independent of the size of the check set.

Campbell demonstrated that codes possessing good soundness naturally exhibit single-shot error correction capabilities, with the specific performance determined by the soundness parameters $t, f$ and the code parameters $[n, k, d]$~\cite{campbell2019theory}. Moreover, good soundness in LDPC codes implies the existence of a macroscopic energy barrier~\cite{aharonov2015quantum,campbell2019theory}, as formalized in the following lemma:

\begin{lemma}[Lemma 3 in~\cite{campbell2019theory}]
\label{soundnesstoenergbarrier}
    Consider a $\left[\left[n, k, d\right]\right]$ quantum code with checks $\mathcal{M}$ that is $(t, f)$-sound and where all qubits are involved in no more than $w_c$ checks. It follows that the energy barrier is at least $f^{-1}(c)$ where $c=\min \left[(t-1) / w_c,\left(d-1\right) / 2\right]$ and $f^{-1}$ is the inverse of the soundness function.
\end{lemma}

It is important to note that the converse statement, any LDPC check family with a macroscopic energy barrier necessarily possesses good soundness, does not hold. The expander code provides a clear counterexample: despite exhibiting bad soundness due to its lack of check redundancy~\cite{ben2010locally}, it has been proven to support single-shot error correction with the small-set flip decoder~\cite{fawzi2020constant,fawzi2018efficient}. In this context, check redundancy refers to the ratio between the check set size (encompassing all checks, including dependent ones) and the number of independent checks.

The soundness of a stabilizer code depends specifically on the selection of the stabilizer set, not merely on the stabilizer generators. Campbell also demonstrated that for any stabilizer code, it is possible to strategically choose a stabilizer set that endows the code with good soundness properties~\cite{campbell2019theory}. The price is that the size of some checks could be excessively large. LDPC property may be lost due to the existence of high-weight checks.

To maintain the LDPC characteristic, Campbell discovered an elegant solution for achieving single-shot properties: applying the hypergraph product construction twice~\cite{campbell2019theory}. The resulting structures naturally possess meta-checks, a direct consequence of the iterative hypergraph product process. Through this approach, Campbell effectively combined soundness and LDPC properties by strategically leveraging check redundancy. The following lemma shows this formally.

\begin{lemma}[Lemma 5 and Lemma 6 in~\cite{campbell2019theory}]
    Let $C_0 \xleftarrow{\partial_0} C_1$ be a chain complex. Applying the hypergraph product, we obtain a new chain complex $\tilde{C}_{-1} \xleftarrow{\tilde{\delta}_{-1}} \tilde{C}_0 \xleftarrow{\tilde{\delta}_0} \tilde{C}_1$, where the maps $\tilde{\delta}_0^T$ and $\tilde{\delta}_{-1}$ are $(t, f)$-sound with $f(x)=x^2 / 4$ and $t=\min \left[d_0, d_0^T\right]$, where $d_0, d_0^T$ are the distances of codes defined by $\delta_0, \delta_0^T$, respectively. Applying the hypergraph products again, we obtain a new chain complex $\breve{C}_{-2} \xleftarrow{\breve{\delta}_{-2}} \breve{C}_{-1} \xleftarrow{\breve{\delta}_{-1}} \breve{C}_0 \xleftarrow{\breve{\delta}_0} \breve{C}_1 \xleftarrow{\breve{\delta}_1} \breve{C}_2$, where the maps $\breve{\delta}_0$ and $\breve{\delta}_{-1}^T$ are $(t, g)$-sound with soundness function $g(x)=x^3 / 4$, and $t=\min \left[d_0, d_0^T\right]$.
\end{lemma}

For a 4D hypergraph product code that constructed with four identical classical codes defined by the same parity check matrix $\delta$ with distance $d$ ($\delta^T$ with distance $d^T$), according to Lemma~\ref{soundnesstoenergbarrier}, the energy barrier for logical operators of the resulting quantum code then has a lower bound of $\Omega\left({(\min[d, d^T])}^{\frac13}\right)$. Later, we will demonstrate that the optimal lower bound for this energy barrier scales as $\Omega(\min[d,d^T])$ in the given scenario.

\subsection{Confinement}

Confinement is another robustness measure of quantum codes. Bomb\'{i}n originally introduced the concept of confinement as a mechanism that allows quantum codes to achieve single-shot error correction~\cite{bombin2015single}. However, in this discussion, we use a more developed definition provided in~\cite{Quintavalle2021single}.

\begin{dfn}[Confinement, Definition 1 in~\cite{Quintavalle2021single}]
   Let $t$ be an integer and $f: \mathbb{Z} \rightarrow$ $\mathbb{Z}$ some increasing function with $f(0)=0$. We say that a stabilizer code has $(t, f)$-confinement if, for all errors $e$ with $\mathrm{wt}_{\mathrm{red}}(e) \leq t$, it holds
   \begin{eqnarray}
       f(|\sigma(e)|) \geq \mathrm{wt}_{\mathrm{red}}(e).
   \end{eqnarray} 
\end{dfn}
In this definition, $\wt_{\mathrm{red}}(e)$ denotes the reduced weight of the error $e$, which represents the smallest weight of any error $e'$ that produces the same error syndrome as $e$. 

Unlike Bomb\'{i}n's more restrictive formulation (Definition 16 in~\cite{bombin2015single}), the above definition permits nonlinear functions $f(x)$. 2D repetition code and 3D hypergraph product codes both exhibit superlinear confinement functions while supporting single-shot correction. These important cases would be excluded under Bomb\'{i}n's original framework.

A code family is said to have good confinement if each code within the family satisfies $(t, f)$-confinement with the following criteria: (1) $t$ increases with the length of the code $n$, specifically $t \in \Omega\left(n^b\right)$ for some $b>0$, and (2) the confinement function $f(\cdot)$ increases monotonically and does not depend on $n$. Additionally, a code family is said to exhibit good X-confinement if this property applies only to Pauli Z errors~\cite{Quintavalle2021single}.

Confinement is a weaker requirement than soundness. While good soundness requires low-weight syndromes to correspond to low-weight errors, good confinement demands that low-weight errors produce low-weight syndromes. For an LDPC code that is $(t,f)$-sound, if the maximum qubit degree is $\omega_q$ (the sparsity parameter), then the code has $({t}/{\omega_q}, f)$ confinement (see Lemma 2 in ~\cite{Quintavalle2021single}). This follows from a direct analysis: if $e$ is an error set with $\wt_{\mathrm{red}}(e) \leq {t}/{\omega_q}$, then its syndrome satisfies $|\sigma(e)| \leq {t}/{\omega_q} \cdot \omega_q = t$. By the soundness of the code, we have $f(|\sigma(e)|) \geq \wt_{\mathrm{red}}(e)$.

Importantly, one can show that codes with good confinement possess single-shot error correction capabilities~\cite{Quintavalle2021single}.

\begin{lemma}[Theorem 1 in \cite{Quintavalle2021single}]
Consider a family of $[[n, k, d]]$ quantum-LDPC codes with good confinement such that $d \geq an^b$ with $a,b>0$. This code family is single-shot for the adversarial noise model. If the code family only has good X-confinement then it is single-shot with respect to Pauli Z noise.
\end{lemma}

Confinement is also related to the energy barrier of a code~\cite{bombin2015single,hong2024quantum,placke2024topological}. In Sec.~\ref{sec:confinement_and_energy_barrier}, we demonstrate formally that confinement naturally induces an energy barrier lower bound.

\subsection{Expander code}
An expander graph embodies a bipartite graph characterized by significant connectivity. Let $G=(V \cup C, E)$ be a bipartite graph. For a vertex subset $S \subseteq V$, a vertex $g \in C$ is a neighbor of $S$ if it is connected to at least one vertex in $S$, $N(S)$ is the set of neighbors of $S$. We say $g \in C$ is a unique neighbor of $S$ if it is only connected to a single vertex in $S$. We use $U(S)$ to denote the collection of unique neighbors of $S$. The expander graph can be defined as follows:

\begin{dfn}[Expander graph]
    Let $G=(V \cup C, E)$ be a bipartite graph with left and right degrees bounded by $w$. Let $|V|=n$ and $|C|=m$. We say $G$ is $(\gamma_1, \partial_1)$-left-expanding with constants $\gamma_1, \partial_1>0$, if for any subset $S \subseteq V$ with $|S| \leq \gamma_1 n$, the neighbor set $N(S)$ of $S$ in $C$ satisfies
    \begin{eqnarray}
        |N(S)| \geqslant w(1-\partial_1)|S|.
    \end{eqnarray}
    Similarly, we say $G$ is $(\gamma_2, \partial_2)$-right-expanding with constants $\gamma_2, \partial_2>0$, if for any subset $T \subseteq C$ with $|T| \leq \gamma_2 n$, the neighbor $N(T)$ of $T$ in the subset $V$ satisfies $|N(T)| \geqslant w(1-\partial_2)|T|$. We call $G$ an $\left(n, m, w, \gamma_1, \partial_1, \gamma_2, \partial_2\right)$ expander graph if it is both left and right expanding.
\end{dfn}

Let $H_G$ be the parity check matrix of an expander code defined by this $\left(n, m, w, \gamma_1, \partial_1, \gamma_2, \partial_2\right)$ expander graph (choosing $V$ and $C$ as the variable and check nodes of the Tanner graph) with $\partial_1<1/2$. Then for any $S \subset V$ such that $|S| \leq \gamma_1 n, \quad U(S) \geq w(1-2 \partial_1)|S|$. Furthermore, the distance $d_C$ of the code $\mathcal{C}(G)$ is greater than $\gamma_1 n$ \cite{Sipser1996}. 

The above analysis directly establishes the confinement property of expander codes. For any error $e$ with reduced weight $\wt_{\mathrm{red}}(e) < d$, we have $|\sigma(e)| \geqslant U(e) \geqslant w(1-2 \partial_1)|e|$. Noting that in classical codes, $\wt_{\mathrm{red}}(e) = |e|$. We can then conclude that $({1}/({w(1-2 \partial_1)}))|\sigma(e)| \geqslant \wt_{\mathrm{red}}(e)$. This relationship demonstrates that expander codes exhibit linear confinement.

Conversely, linear confinement also induces left-expansion. For $(t, f)$-confinement with $t = an$ and $f(x) = cx$, where $a$ and $c$ are constants, any error $S$ with $|S| < an$ satisfies $c|\sigma(S)| \geqslant |S|$. Thus $N(S) \geqslant |\sigma(S)| \geqslant (1/c)|S|$. With degree $w$, we have $N(S) \geqslant w(1-(1-{1}/{(cw)}))|S|$. Therefore, this code is an expander code with $(a, 1-{1}/{(cw)})$-left-expanding.

Furthermore, the energy barrier of expander codes can be characterized by the following lemma.
\begin{lemma}
    Let $\Delta(H_G)$ be the energy barrier of classical expander code defined by $H_G$ with parameters $(n, m, w, \gamma_1, \partial_1, \gamma_2, \partial_2)$ and $\partial_1<{1}/{2}$, then
    \begin{eqnarray}
        \Delta(H_G) \geq cn,
    \end{eqnarray}
    where $c > 0$ is a constant.
    \label{lemma:energy_barrer_expander_code}
\end{lemma}
\begin{proof}
    Without loss of generality, suppose that energy barrier $\Delta(H_G)$ of code is given by the energy barrier of logical operator $L$, with path $r_L$. Since $|L| \geq\gamma_1 n$, there is a point along the path (denoted as $P$) such that $\left|P\right|=\gamma_1 n$. With the expansion property $U(P) \geq w(1-2 \partial_1)|P|$, we get $\Delta(H_G) \geq U(P) \geq w(1-2 \partial_1) |P| = w(1-2 \partial_1) \gamma_1 n$. Then, we have $\Delta(H_G) \geq cn$ for $c = w(1-2 \partial_1) \gamma_1 >0$.
\end{proof}

From Lemma~\ref{lemma:energy_barrer_expander_code}, it follows that the classical expander code has an energy barrier scaling linearly with the number of bits. A similar conclusion holds for the transpose code in which the variable and the check nodes are exchanged. Therefore, according to the results in~\cite{zhao2024on}, the quantum expander code defined by parity check matrix $H_{(H_G, H_G)}$ has $\Omega(\sqrt{n})$ energy barrier, where $n$ here is the number of qubits.

Researches have shown that good soundness requires check redundancy in LDPC codes~\cite{ben2010locally,ben2003stoc}, the expander code exhibits poor soundness due to its lack of check redundancy. However, despite this limitation, the quantum expander code has been proven to support single-shot error correction when implemented with the small-set flip decoder~\cite{fawzi2020constant,fawzi2018efficient}. This makes confinement a more general and inclusive property compared to soundness, in relation to single-shot error correction.

\subsection{Higher-dimensional hypergraph product}

The hypergraph products provide a way to construct quantum codes with any two classical linear codes~\cite{2014-Tillich}. It enables the application of extensive insights from classical coding theory to quantum coding. Those codes have good properties, such as circuit-level distance preservation~\cite{manes2023distance,tan2025effective}.

Given two classical codes $C_1$ and $C_2$, with respective parity check matrices $H_1$ and $H_2$, the hypergraph product defines a quantum code with parity check matrices
\begin{equation}
\begin{aligned}
H_X & =\left(H_1 \otimes I \quad I \otimes H_2^T\right), \\
H_Z & =\left(I \otimes H_2 \quad H_1^T \otimes I\right).
\end{aligned}
\end{equation}
One can verify that $H_X H_Z^T=0$ within the field $\mathbb{F}_2$. Thus these two parity check matrices define a CSS code, with a quantum parity check matrix of the form
\begin{equation}
H_{\left(H_1, H_2\right)}=\left(\begin{array}{cc}
H_X & 0 \\
0 & H_Z
\end{array}\right).
\end{equation}

Hypergraph product codes can be interpreted through chain complexes. Classical codes correspond to length-2 chain complexes, where boundary maps define parity check matrices. A hypergraph product code arises from the tensor product of two such chain complexes, yielding a length-3 chain complex. This structure provides two boundary maps that serve as parity checks. The fundamental property that ``a boundary has no boundary'' ensures these checks commute, allowing the construction of CSS quantum codes directly from the boundary maps of the resulting length-3 chain complex.

HHGP codes~\cite{zeng2019higher} extends this concept by generalizing the construction to create a length-$n$ chain complex. This generalization has important connections to topological codes: $n$-dimensional toric codes on hypercubic lattices constitute a length-$n$ chain complex~\cite{bombin2007homological,dennis2002topological}, and can be perfectly described by higher-dimensional hypergraph products.

Following Ref.~\cite{zeng2019higher}, we give a short introduction to the chain complex interpretation of hypergraph product. A chain complex is a sequence of finite-dimensional vector spaces $\{ \ldots, \mathcal{A}_{j-1}, \mathcal{A}_j, \ldots \}$ connected by boundary operators $\partial_j: \mathcal{A}_{j-1} \leftarrow \mathcal{A}_{j}$, satisfying $\partial_j \partial_{j+1}=0$ for all $j \in \mathbb{Z}$. 

For qubit systems, we consider vector spaces $\mathcal{A}_j=\mathbb{F}_2^{n_j}$, consisting of binary vectors of length $n_j$. $\mathcal{A} \equiv \mathcal{K}\left(A_1, \ldots, A_m\right)$ can be defined as a length-$(m+1)$ chain complex with boundary map $\partial_j$ represented by $n_{j-1} \times n_j$ binary matrices $A_j$:
\begin{eqnarray}
    \mathcal{A}:\{0\} \stackrel{\partial_0}{\leftarrow} \mathcal{A}_0 \stackrel{A_1}{\leftarrow} \mathcal{A}_1 \ldots \stackrel{A_m}{\leftarrow} \mathcal{A}_m \stackrel{\partial_{m+1}}{\leftarrow}\{0\},
\end{eqnarray}
where adjacent matrices satisfy the orthogonality condition $A_{j-1} A_j=0$ for $j \in\{1, \ldots, m\}$. The trivial operators $\partial_0:\{0\} \leftarrow \mathcal{A}_0$ and $\partial_{m+1}: \mathcal{A}_m \leftarrow\{0\}$ are treated as zero matrices of dimensions $0 \times n_0$ and $n_m \times 0$, respectively.

The subspace $\operatorname{Im}\left(A_{j+1}\right) \subseteq \mathcal{A}_j$ consists of boundaries, which can be represented as linear combinations of the columns of $A_{j+1}$. This subspace defines a binary linear code with generator matrix $A_{j+1}^T$, denoted as $\operatorname{Im}\left(A_{j+1}\right)=\mathcal{C}_{A_{j+1}^T}$. For the special case where $j=m$, we have $\operatorname{Im}\left(\partial_{m+1}\right)={0}$, which forms a trivial vector space.

The subspace $\operatorname{Ker}\left(A_j\right) \subseteq \mathcal{A}_j$ consists of cycles, which are vectors $x \in \mathcal{A}_j$ orthogonal to the rows of $A_j$, satisfying $A_j x^T=0$. This subspace defines a binary linear code with parity check matrix $A_j$, expressed as $\operatorname{Ker}\left(A_j\right)=\mathcal{C}_{A_j}^{\perp}$. In the special case where $j=0$, we have $\operatorname{Ker}\left(\partial_0\right)=\mathcal{A}_0$, encompassing the entire vector space.

Within the binary vector space $\mathcal{A}_j$, we have defined two key subspaces: $\operatorname{Im}\left(A_{j+1}\right)$, consisting of boundaries, and $\operatorname{Ker}\left(A_j\right)$, comprising cycles. A fundamental property of chain complexes is that all boundaries are cycles, though not all cycles are boundaries. This relationship is mathematically expressed through the orthogonality condition $A_j A_{j+1}=0$.

Based on these subspaces, we can define the $j$-th homology group as
\begin{eqnarray}
H_j(\mathcal{A}) \equiv \operatorname{Ker}\left(A_j\right) / \operatorname{Im}\left(A_{j+1}\right).
\end{eqnarray}
This quotient space contains all elements that are cycles but not boundaries, which we refer to as homologically nontrivial cycles.

The rank of the $j$-th homology group, representing the dimension of the associated vector space, is given by
\begin{eqnarray}
k_j \equiv \operatorname{rank} H_j(\mathcal{A})=n_j-\operatorname{rank} A_j-\operatorname{rank} A_{j+1}.
\end{eqnarray}

Additionally, the homological distance $d_j$ is defined as the minimum Hamming weight of a nontrivial element in $H_j(\mathcal{A})$, expressed as:
\begin{eqnarray}
d_j=\min \{\operatorname{wt}(x): x \in H_j(\mathcal{A}), x \neq 0\}.
\end{eqnarray}

A quantum CSS code can be derived from a length-3 chain complex $\mathcal{A}_0 \stackrel{A_1}{\leftarrow} \mathcal{A}_1 \stackrel{A_2}{\leftarrow} \mathcal{A}_2$ by setting $H_Z^T=A_1$ and $H_X=A_2$. In chain complex terminology, the code length equals $\operatorname{dim} \mathcal{A}_1$, while its dimension corresponds to the first homology group $\mathcal{H}_{1}=\operatorname{ker} A_{1} / \operatorname{Im} A_2$, or equivalently, the first cohomology group $\mathcal{H}_{1}^*=\operatorname{ker} A_1^T / \operatorname{Im} A_{2}^T$. The X and Z distances are determined by the minimum weights of non-zero vectors in $\mathcal{H}_{1}^*$ and $\mathcal{H}_{1}$ respectively.

Classical codes can be described by length-2 chain complexes. To construct quantum codes from classical codes, it is just a question of how to build a length-3 chain complex from two length-2 chain complexes. For this purpose, we employ the tensor product of chain complexes.

The tensor product $\mathcal{A} \times \mathcal{B}$ of two chain complexes $\mathcal{A}$ and $\mathcal{B}$ is defined as a new chain complex in the following way
\begin{eqnarray}
    (\mathcal{A} \times \mathcal{B})_l=\bigoplus_{i+j=l} \mathcal{A}_i \otimes \mathcal{B}_j.
\end{eqnarray}
The vector spaces in the new chain complex are direct sums of Kronecker products of vector spaces from the original chain complexes. Moreover, the boundary operators in this tensor product construction act as
\begin{eqnarray}
    \partial_{i+j}(a \otimes b)=\partial_i^{\prime} a \otimes b+(-1)^i a \otimes \partial_j^{\prime \prime} b,
\end{eqnarray}
where $a \in \mathcal{A}_i, b \in \mathcal{B}_j$, and $\partial_i^{\prime}$ and $\partial_j^{\prime \prime}$ are the boundary operators in the corresponding space $\mathcal{A}$ and $\mathcal{B}$. If both $\mathcal{A}$ and $\mathcal{B}$ have finite dimension, then the dimension of $\mathcal{C}=\mathcal{A} \times \mathcal{B}$ is
\begin{eqnarray}
    n_j(\mathcal{C})=\sum_i n_i(\mathcal{A}) n_{j-i}(\mathcal{B}).
\end{eqnarray}

By the K\"unneth theorem, the homology groups of the product complex $\mathcal{C}=\mathcal{A} \times \mathcal{B}$ are given by
\begin{eqnarray}
    H_j(\mathcal{C}) \cong \bigoplus_i H_i(\mathcal{A}) \otimes H_{j-i}(\mathcal{B}) .
\end{eqnarray}
Consequently, the rank $k_j(\mathcal{C})$ of the $j$-th homology group $H_j(\mathcal{C})$ can be expressed as
\begin{eqnarray}
    k_j(\mathcal{C})=\sum_i k_i(\mathcal{A}) k_{j-i}(\mathcal{B}) .
\end{eqnarray}

The hypergraph product code is defined as the tensor product of two length-2 chain complexes. The higher-dimensional hypergraph product extends this concept to the tensor product of two chain complexes of arbitrary lengths. As demonstrated in Ref.~\cite{zeng2019higher}, for HHGP codes formed from a length-m chain complex $\mathcal{A}$ and a length-l chain complex $\mathcal{B}$ with $l = 2$, there exists a lower bound on the distance of logical operators in the resulting quantum codes:
\begin{eqnarray}
    d_j(\mathcal{A} \times \mathcal{B})=\min \left[d_{j-1}(\mathcal{A}) d_1(\mathcal{B}), d_j(\mathcal{A}) d_0(\mathcal{B})\right].
    \label{eq:hhpdistance}
\end{eqnarray}

In this work, we focus specifically on those 3D and 4D hypergraph product codes that can be regarded as tensor products of three or four classical codes, respectively. This allows us to apply the distance bounds given by above equation.

\section{Confinement and energy barrier}
\label{sec:confinement_and_energy_barrier}

The relation between confinement and energy barrier is well-known~\cite{bombin2015single,hong2024quantum,placke2024topological}. In this section, we formalize this connection through the following lemma.
\begin{lemma}
\label{lemma:confinement_energybarrier}
    For a $[[n, k, d]]$ quantum code $\mathcal{C}$ with $(t,f)$ confinement where $t \leqslant d$ and $f(\cdot)$ is increasing monotonically, the energy barrier of the code is at least $f^{-1}(w)$, where $w = \min[t, {(d-1)}/{2}]$.
\end{lemma}
\begin{proof}
    Without loss of generality, suppose that the code energy barrier $\Delta(\mathcal{C})$ is given by logical operator $L$, with a path $r_L$. Since $|L| \geqslant d \geqslant t$, we first show that there always exist a step $P_c$ along the path $r_L$, such that $\wt_{\mathrm{red}}(P_c)\geqslant (d-1)/2$. This can be argued as follows~\cite{campbell2019theory}:
    
    Suppose $r_L = \{P_0, P_1, \cdots, P_L\}$, and each $E_j = P_j P_{j-1}$ is single qubit Pauli operator. For each $P_j$, we define the reduced weight as:
    \begin{eqnarray}
        \mathrm{wt}_{\mathrm{red}}(P_j):=\min_V\{\mathrm{wt}(P_j V): V \in \mathcal{P}, \sigma(V)=0\},
    \end{eqnarray}
    where the minimization is over Pauli operators $V$ with trivial syndrome. Let $V_j$ denote the Pauli operators achieving this minimum for $P_j$. Since $\sigma\left(V_j\right)=0$, $V_j$ is either a stabilizer or a nontrivial logical operator.
    
    Trivially, we find $V_0=I$ and $V_L=P_L$, indicating that the sequence starts with a stabilizer and ends with a logical operator. Therefore, there exists an index $j^*$ such that $V_{j^*}$ is a stabilizer, while $V_{j^*+1}$ is a nontrivial logical operator. Then $V_{j^*} V_{j^*+1}$ is also a nontrivial logical operator, leading to:
    \begin{eqnarray}
        d \leq \mathrm{wt}\left(V_{j^*} V_{j^*+1}\right) .
    \end{eqnarray}
    Notice that
    \begin{eqnarray}
        \operatorname{wt}\left(V_{j^*} V_{j^*+1}\right)&=&\operatorname{wt}\left(V_{j^*} V_{j^*+1} P_{j^*}P_{j^*}^\dag P_{j^*+1} P_{j^*+1}^\dag\right) \nonumber \\
        &=&\operatorname{wt}\left(V_{j^*}P_{j^*} V_{j^*+1} P_{j^*+1} P_{j^*}^\dag P_{j^*+1}^\dag\right). \nonumber \\
    \end{eqnarray}
    Applying the triangle inequality twice gives:
    \begin{eqnarray}
        \mathrm{wt}\left(V_{j^*} V_{j^*+1}\right) \leq \mathrm{wt}_{\mathrm{red}}\left(P_{j^*}\right)+\mathrm{wt}_{\mathrm{red}}\left(P_{j^*+1}\right)+1 .
    \end{eqnarray}
    
    Thus, combining this result, we obtain:
    \begin{eqnarray}
        d \leq 2 \cdot \max \left[\mathrm{wt}_{\mathrm{red}}\left(P_{j^*}\right), \mathrm{wt}_{\mathrm{red}}\left(P_{j^*+1}\right)\right]+1,
    \end{eqnarray}
    which implies
    \begin{eqnarray}
        \frac{d-1}{2} \leq \max \left[\mathrm{wt}_{\mathrm{red}}\left(P_{j^*}\right), \mathrm{wt}_{\mathrm{red}}\left(P_{j^*+1}\right)\right].
    \end{eqnarray}
    Thus, the sequence of reduced weights $\left\{\mathrm{wt}_{\mathrm{red}}\left(P_0\right), \mathrm{wt}_{\text {red}}\left(P_1\right), \cdots, \mathrm{wt}_{\text {red }}\left(P_L\right)\right\}$ starts and ends with zero, and it reach a value of at least $\left(d-1\right) / 2$ in the middle. Moreover, the local error condition ensures that $\left|\mathrm{wt}_{\mathrm{red}}\left(P_{j+1}\right)-\mathrm{wt}_{\text {red }}\left(P_j\right)\right|$ is either 0 or 1, then the sequence must include every integer from 0 to $\left(d-1\right) / 2$.

    Recall that $(t,f)$ confinement means that for all errors $e$ with $\wt_{\mathrm{red}}(e)\leq t$, $f(|\sigma(e)|) \geq \wt_{\mathrm{red}}(e)$. Then for $w = \min [t, (d-1)/2]$, there must exist an $P_w$ with $\wt_{\text {red}}(P_w) = w$, such that $f(|\sigma(P_w)|) \geqslant \wt_{\mathrm{red}}(P_w) = w$. Given $f(\cdot)$ is monotonically increasing, $ |\sigma(P_c)| \geqslant f^{-1}(w)$. This implies that the energy barrier of code $\mathcal{C}$ is at least $f^{-1}(w)$.
\end{proof}

The converse statement of Lemma~\ref{lemma:confinement_energybarrier}—that any quantum code with a macroscopic energy barrier necessarily exhibits confinement—does not generally hold. This can be demonstrated with a straightforward counterexample. Consider a quantum LDPC code $\mathcal{C}$ constructed as a composite system comprising two regions: Region A, based on a 2D quantum repetition code $[[n,1,n]]$ with energy barrier $O(\sqrt{n})$, and Region B, based on a 1D quantum repetition code $[[n,1,n]]$ with constant energy barrier $O(1)$, see Fig.~\ref{fig:double_ising_model}. The overall code is a $[[2n,1,2n]]$ quantum repetition code with energy barrier $O(\sqrt{n})$ determined by Region A. However, for an error $e$ contained entirely within Region B, we can have $\wt_{\mathrm{red}}(e) = |e|$ that grows up to $(d-1)/2$ with $|\sigma(e)| \leq c$, where $c$ is constant. Consequently, the overall code fails to satisfy the confinement criterion.

Quintavalle et al. proved that all 3D hypergraph product codes exhibit $(t, f)$ X-confinement (according to their convention) where $t$ equals the minimum distance of classical codes and $f(x)= x^3/4$ (see Lemma 10 of Ref~\cite{Quintavalle2021single}). For a 3D hypergraph product code constructed with three identical classical codes defined by the same parity check matrix $\delta$ with distance $d$, the result in Ref~\cite{Quintavalle2021single} implies a lower bound for the energy barrier of Z-type logical operators as $\Omega\left(d^{\frac{1}{3}}\right)$. In section~\ref{sec:energy_barrier_of_3d_hgp}, we will demonstrate that the optimal lower bound for this energy barrier actually scales as $\Omega(d)$ in this scenario.

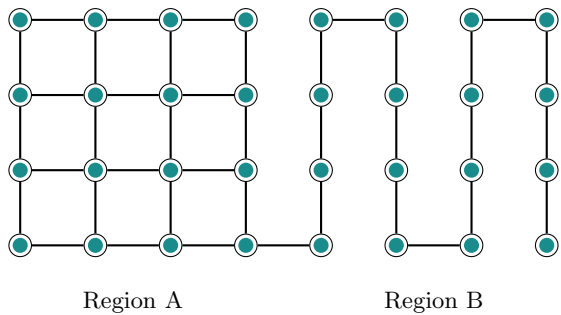
\begin{figure}
\centering
\begin{tikzpicture}[
    spin/.style={circle, draw, minimum size=0.3cm, inner sep=0pt},
    connection/.style={-,thick},
    node distance=1cm
]

% 2D Ising model (left side)
\foreach \i in {0,...,3} {
    \foreach \j in {0,...,3} {
        \node[spin] (s\i\j) at (\i,\j) {};
        
        % Fill with up or down spin randomly
        \pgfmathparse{rnd > 0.0 ? 1 : -1}
        \pgfmathtruncatemacro{\spinvalue}{\pgfmathresult}
        \ifnum\spinvalue=1
            \fill[teal, opacity=0.9] (s\i\j) circle (0.1cm);
        \fi
    }
}

% Horizontal connections in 2D model
\foreach \i in {0,...,2} {
    \foreach \j in {0,...,3} {
        \pgfmathtruncatemacro{\nexti}{\i+1}
        \draw[connection] (s\i\j) -- (s\nexti\j);
    }
}

% Vertical connections in 2D model
\foreach \i in {0,...,3} {
    \foreach \j in {0,...,2} {
        \pgfmathtruncatemacro{\nextj}{\j+1}
        \draw[connection] (s\i\j) -- (s\i\nextj);
    }
}

% 1D Ising model (right side) - arranged in a snake pattern
% Define positions for the 1D model
\foreach \i in {0,...,3} {
    \foreach \j in {0,...,3} {
        \pgfmathtruncatemacro{\idx}{4*\i + \j}
        \node[spin] (r\idx) at (\i+4, \j) {};
        
        % Fill with up or down spin randomly
        \pgfmathparse{rnd > 0.0 ? 1 : -1}
        \pgfmathtruncatemacro{\spinvalue}{\pgfmathresult}
        \ifnum\spinvalue=1
            \fill[teal, opacity=0.9] (r\idx) circle (0.1cm);
        \fi
    }
}

% Connections in 1D model (snake pattern)
% First row (left to right)
\draw[connection] (r0) -- (r1);
\draw[connection] (r1) -- (r2);
\draw[connection] (r2) -- (r3);
% Connect to second row
\draw[connection] (r3) -- (r7);
% Second row (right to left)
\draw[connection] (r7) -- (r6);
\draw[connection] (r6) -- (r5);
\draw[connection] (r5) -- (r4);
% Connect to third row
\draw[connection] (r4) -- (r8);
% Third row (left to right)
\draw[connection] (r8) -- (r9);
\draw[connection] (r9) -- (r10);
\draw[connection] (r10) -- (r11);
% Connect to fourth row
\draw[connection] (r11) -- (r15);
% Fourth row (right to left)
\draw[connection] (r15) -- (r14);
\draw[connection] (r14) -- (r13);
\draw[connection] (r13) -- (r12);

% Connection between 2D and 1D models
\foreach \j in {0} {
    \draw[connection] (s3\j) -- (r\j);
}

% Labels
\node[below=0.5cm] at (1.5,-0.0) {Region A};
\node[below=0.5cm] at (5.5,-0.0) {Region B};

\end{tikzpicture}
\caption{A composite system consisting of a 2D quantum repetition code (Region A) and a 1D quantum repetition code arranged in a snake pattern (Region B). The overall code exhibits an energy barrier determined by the 2D quantum repetition code, which scales as the linear dimension of Region A. However, confinement properties are influenced by the 1D quantum repetition component, resulting in the absence of confinement for the composite system.}
\label{fig:double_ising_model}
\end{figure}

\section{Tensor product code and its energy barrier}
\label{sec:tensor_product_code_and_its_energy_barrier}

In this section, we analyze the energy barrier of tensor product codes — a specific classical code structure that builds new classical codes from two existing ones.  We will demonstrate that the energy barrier of the resulting code is related to the distance of the underlying classical codes.

This result is essential for our subsequent analysis of HHGP codes. Specifically, we establish that the logical operators of HHGP codes admit a tensor product decomposition. This structural property enables us to derive a lower bound on the energy barrier of HHGP codes in terms of the distances of their constituent classical codes.

\subsection{Code construction}

Given two classical codes $\mathcal{C}_a$ and $\mathcal{C}_b$ with parity check matrix $\delta_a$ and $\delta_b$, the tensor product code $\mathcal{C}_c$  is a classical code defined to have the parity check matrix
\begin{eqnarray}
    \delta_c = \left(\begin{array}{c}
         \delta_a \otimes I_{n_b} \\
         I_{n_a} \otimes \delta_b
    \end{array}\right).
\end{eqnarray}

This tensor product structure can be also understood from the perspective of code space. Let $\mathcal{C}_a$ with parameters $\left[n_a, k_a, d_a\right]$ and $\mathcal{C}_b$ with parameters $\left[n_b, k_b, d_b\right]$ be two classical codes. The tensor product of these two codes, denoted as $\mathcal{C}_c = \mathcal{C}_a \otimes \mathcal{C}_b$, encodes messages as follows: The message of the tensor product code $\mathcal{C}_a \otimes \mathcal{C}_b$ can be represented as a $k_a \times k_b$ matrix $M$. The encoding proceeds in two stages: first, each column of $M$ is encoded using $\mathcal{C}_a$, yielding a $n_a \times k_b$ intermediate matrix; second, each row of this intermediate matrix is encoded using $\mathcal{C}_b$, producing the final $n_a \times n_b$ codeword.

Let $L_a$ be the logical operator of $\delta_a$, and $L_b$ be the logical operator of $\delta_b$, one can verify that $L_c = L_a \otimes L_b$ is the logical operator of $\delta_c$ because
\begin{eqnarray}
    \delta_c L_c &=& \left(\begin{array}{c}
         \delta_a \otimes I_{n_b} \\
         I_{n_a} \otimes \delta_b
    \end{array}\right) (L_a \otimes L_b) \nonumber \\
    &=&  \left(\begin{array}{c}
         (\delta_a \otimes I_{n_b}) (L_a \otimes L_b) \\
         (I_{n_a} \otimes \delta_b) (L_a \otimes L_b)
    \end{array}\right)  \nonumber \\
    &=& 0.
\end{eqnarray}
Here we used the relation $\delta_aL_a = \delta_bL_b = 0$. 

Given $k_a$ distinct logical operators $L_a$ and $k_b$ distinct logical operators $L_b$, one can construct $k_ak_b$ distinct logical operators $L_c$ through this tensor product construction.

Moreover, the distance of the tensor product code is at least $d_a d_b$. A simple argument demonstrates this: consider two different message matrices of code $\mathcal{C}_c$, denoted as $M_a$ and $M_b$. After the first encoding step, we obtain intermediate matrices $N_a$ and $N_b$, and the second encoding step produces the final codewords $C_a$ and $C_b$. If $M_a$ and $M_b$ differ in the $i$-th row, then $N_a$ and $N_b$ differ in at least $d_a$ positions in that row, corresponding to $d_a$ columns. Let the indices of these differing columns be $\{j_1, \ldots, j_{d_a}\}$. During the second encoding step, each of these columns generates differences in at least $d_b$ positions. Consequently, $C_a$ and $C_b$ differ in at least $d_a d_b$ positions. In summary, the tensor product code $\delta_c$ has parameters $[n_an_b, k_ak_b, d_ad_b]$.

\subsection{Energy barrier of tensor product code}
\label{sub:energy_barrier_tensor_product}

Let $L_a$ and $L_b$ be the logical operators of $\delta_a$ and $\delta_b$. Our goal is to determine the energy barrier of the operator $L_c = L_a \otimes L_b$, which serves as a logical operator for $\delta_c$.

To compute the energy barrier of $L_c$, we consider a path $r = \{P_0, P_1, \cdots, P_F\}$ with $P_0 = I$ and $P_F=L_c$. Then
\begin{equation}
    \Delta(L_c) = \min \{ \epsilon_{\max}(r): r \in w(0, L_c) \}.
\end{equation}
Note that $r$ is a path on the bit of code $\delta_c$. Given a $P_\ell \in r$, the energy (the number of violated checks) of $P_\ell$ is
\begin{eqnarray}
    \epsilon(P_\ell) &=& \wt(\delta_cP_\ell) \nonumber \\
    &=& \wt \left(\left(\begin{array}{c}
         \delta_a \otimes I_{n_b} \\
         I_{n_a} \otimes \delta_b
    \end{array}\right) P_\ell \right) \nonumber \\
    &=& \wt((\delta_a \otimes I_{n_b}) P_\ell) + \wt((I_{n_a}\otimes\delta_b)P_\ell).
    \label{eq:energy_pell_main}
\end{eqnarray}

Because $\wt(\cdot) \geqslant 0$, then we have 
\begin{eqnarray}
    \epsilon(P_\ell) \geqslant \wt((\delta_a \otimes I_{n_b}) P_\ell),
\end{eqnarray}
or 
\begin{eqnarray}
    \epsilon(P_\ell) \geqslant \wt((I_{n_a}\otimes\delta_b)P_\ell).
\end{eqnarray}

Because the path $r$ contains paths for $L_a$ and $L_b$, the above inequalities indicate that the energy barrier of $L_c$ is lower bounded by the energy barrier of $L_a$ or $L_b$. Furthermore, the energy barrier of code $\delta_c$ is then lower bounded by energy barrier of codes $\delta_a$ and $\delta_b$. We conclude this as the following Lemma
\begin{lemma}
\label{lemma:energy_barrier_tensor_product_energy}
    Given two classical codes $\delta_a$ and $\delta_b$ with parameters $[n_a, k_a, d_a, E_a]$ and $[n_b, k_b, d_b, E_b]$, where $E_a$ and $E_b$ are energy barriers of codes. Then the energy barrier of the tensor product code $\delta_c$ is lower bounded by
    \begin{eqnarray}
        \Delta(\delta_c) \geqslant \min [E_a, E_b].
    \end{eqnarray}
\end{lemma}

However, this may not be the tightest bound, particularly when codes $\delta_a$ and $\delta_b$ have constant energy barriers. In the following, we establish a relationship between the energy barrier of $L_c$ and the distance of $L_a$ and $L_b$.

\subsubsection{Lower bound}
Given two classical codes $\delta_a$ and $\delta_b$ with energy barriers $\Delta(\delta_a) = E_a$ and $\Delta(\delta_b) = E_b$. Here, we show that the energy barrier of the tensor product code $\delta_c$ can be lower bounded by the distances of the two constituent codes. Formally, we state the following lemma.

\begin{lemma}
\label{lemma:energy_barrier_of_tensor_product_codes}
    Given two classical codes $\delta_a$ and $\delta_b$ with parameters $[n_a, k_a, d_a, E_a]$ and $[n_b, k_b, d_b, E_b]$. The energy barrier of the tensor product code $\delta_c$ is lower bounded by
    \begin{eqnarray}
        \Delta(\delta_c) \geqslant \min [d_a, d_b].
    \end{eqnarray}
\end{lemma}

\begin{proof}
    Let us first establish a clear representation of the tensor product code structure. Given the parity check matrices $\delta_a \in \{0,1\}^{m_a \times n_a}$ and $\delta_b \in \{0,1\}^{m_b \times n_b}$, we can arrange the code words of $\delta_c$ on a two-dimensional grid of size $n_a \times n_b$. Then, each row of the 2D grid corresponds to a copy of code $\delta_b$, and each column corresponds to a copy of code $\delta_a$.
    
    The logical operator $L_c = L_a \otimes L_b$ can be visualized as a pattern on this grid, where $L_a$ specifies which rows are involved and $L_b$ specifies which columns are involved. Specifically, the support of $L_c$, denoted as $\mathrm{supp}(L_c)$, consists of positions $(i,j)$ where $i \in \mathrm{supp}(L_a)$ and $j \in \mathrm{supp}(L_b)$, forming a subgrid of size $|L_a| \cdot |L_b| = d_a d_b$.
    
    Any valid path $r = \{P_r^0, P_r^1, P_r^2, \ldots, P_r^L\}$ from the identity operator $I$ to the logical operator $L_c$, by definition, must flip all the bits in the subset $\mathrm{supp}(L_c)$. For any intermediate state $P_r^i$ along this path, the syndrome weight $\wt(\delta_c P_r^i)$ measures the energy cost at that step. To finish the proof, We show that for any valid path, there must exist a intermediate state with syndrome weight at least $\min [d_a, d_b]$.

    To flip all bits in $\mathrm{supp}(L_c)$, the path must eventually form a logical operator in each of the $d_a$ rows indexed by $\mathrm{supp}(L_a)$, as well as in each of the $d_b$ columns indexed by $\mathrm{supp}(L_b)$. More precisely, in each row $i \in \mathrm{supp}(L_a)$, the final configuration must contain a codeword of $\delta_b$. Similarly, each column $j \in \mathrm{supp}(L_b)$ must contain a codeword of $\delta_a$.

    Since the path proceeds by single bit flips, there must exist a critical intermediate state $P^{i^*}$, which represents the first moment when a logical operator of either $ \delta_a $ or $ \delta_b $ appears. More precisely, $P^{i^*}$ could be one of the following configurations:
    
    \begin{enumerate}
        \item Exactly one row $i \in \text{supp}(L_a)$ contains a non-trivial logical operator $L'_b$ of code $\delta_b$ (not necessarily to be $L_b$), while other rows and columns contain none, or
        \item Exactly one column $j \in \text{supp}(L_b)$ contains a non-trivial logical operator $L'_a$ of code $\delta_a$ (not necessarily to be $L_a$), while other rows and columns contain none, or
        \item In rare cases, exactly one row $i \in \text{supp}(L_a)$ contains a non-trivial logical operator $L'_b$ of code $\delta_b$ (not necessarily to be $L_b$), while simultaneously, exactly one column $j \in \text{supp}(L_b)$ contains a non-trivial logical operator $L'_a$ of code $\delta_a$ (not necessarily to be $L_a$), with all other rows and columns containing none.
    \end{enumerate}

    \textit{Case 1:} The row containing the logical operator has $d_b$ flipped bits at positions $(i, j_1), \ldots, (i, j_{d_b'})$ where $\{j_1, \ldots, j_{d_b'}\} \subseteq \mathrm{supp}(L_b')$. Each of these bits lies in a distinct column. Since these columns do not contain complete codewords of $\delta_a$, each such incomplete column contributes at least 1 to the syndrome weight. With $d_b' \geqslant d_b$, the total syndrome weight is at least $d_b$.

    \textit{Case 2:} Similarly, the column containing the logical operator $L_a'$ has $d_a'$ flipped bits, each in a distinct row that does not form a complete codeword of $\delta_b$. Thus each incomplete row contributes at least 1 to the syndrome weight, yielding a total contribution of at least $d_a$.

    \textit{Case 3:} The row containing a logical operator of $\delta_b$ contributes at least $d_b$ flipped bits, and the column containing a logical operator of $\delta_a$ contributes at least $d_a$ flipped bits. With out loss of generality, suppose these two logical operators intersect at position $(i,j)$. The row induces syndrome violations in $d_b - 1$ columns (excluding the column that contains a logical operator), and the column induces violations in $d_a - 1$ rows (excluding the row that contains a logical operator). Thus, the syndrome weight is at least $(d_a - 1) + (d_b - 1) = d_a + d_b - 2 \geq \min[d_a, d_b]$, where the final inequality holds when $d_a, d_b \geq 2$.

    In all cases, provided $d_a, d_b \geq 2$ (which is standard for codes with non-trivial distance), the syndrome weight at the critical state $P^{i^*}$ is at least $\min [d_a, d_b]$. Since the energy barrier $\Delta(L_c)$ is defined as the minimum over all paths of the maximum syndrome weight along each path, we conclude $\Delta(L_c) \geq \min[d_a, d_b].$

    The above result holds for any logical operator $L_c$ of $\delta_c$. Let us denote the set of logical operators of $\delta_c$ as $L(\delta_c)$, then we have
    \begin{equation}
    \Delta(\delta) = \min\{\Delta(L_c): L_c \in L(\delta_c)\} \geq \min[d_a, d_b].
    \end{equation}
\end{proof}

The intuition of Lemma~\ref{lemma:energy_barrier_of_tensor_product_codes} is from the 2D repetition code, which can be interpreted as the tensor product of two 1D repetition code. The 1D repetition code has parameters $[L, 1, L, \Omega(1)]$, where the energy barrier scales as $O(1)$. In contrast, the 2D repetition code exhibits parameters $\left[L^2, 1, L^2, \Omega(L)\right]$, with an energy barrier scaling as $\Omega(L)$.

The energy barrier of the 2D repetition code can be rigorously analyzed through its tensor product structure. Consider a 2D repetition code on an $L \times L$ square lattice, where the logical operator corresponds to flipping all bits on the lattice. For any valid path (meets local error condition) implementing this logical operator, there necessarily exists an intermediate configuration where precisely a single row or a single column of bits has been flipped  (or, a single row and a single column simultaneously). At this configuration, the number of violated checks (i.e., the energy cost) is at least $L$, as verified by analyzing the violated checks on each column and row. Consequently, the energy barrier of the 2D repetition code is lower-bounded by $L$.

\subsubsection{Upper bound - ``strip'' argument}

To establish an upper bound, one can just consider any specific path from the identity to a logical operator. We analyze two distinct implementation strategies: row-wise and column-wise implementation of logical operators. Same to the ``strip'' argument in Ref.~\cite{rakovszky2024physicsgoodldpccodes}.

We still view the tensor product code $\delta_c$ as operating on a 2D grid of size $n_a \times n_b$. In this representation, each row corresponds to a copy of code $\delta_b$, each column corresponds to a copy of code $\delta_a$.

For the row-wise implementation, let $r_a = \{P_a^0, P_a^1, \ldots, P_a^F\}$ represent an optimal path for implementing logical operator $L_a$, where $P_a^0 = I$ and $P_a^F = L_a$. By definition, the maximum syndrome weight along the path $r_a$ is $E_a$. Because $P_iP_{i+1} = E_i$ is a single Pauli for all $i$, we denote $q_a = \{j_1, j_2, \ldots, j_{L_a}\}$ as the position of single Pauli in each step of $r_a$. These positions correspond to rows on the 2D grid of the tensor product code.

We consider a path for implementing $L_c$ by applying $L_b$ sequentially according to the positions in $q_a$. Specifically, Our path $r_c$ for implementing $L_c$ begins with $P_c^0 = I$, for each position $j_k \in q_a$, we implement the logical operator $L_b$ on the corresponding row. After finishing all rows in $q_a$, we have $P_c^F = L_c$.

Now consider the energy barrier given by this specific $r_c$. Given any $P_c^i \in r_c$, the energy of $P_c^i$ is
\begin{eqnarray}
   \epsilon(P_c^i) = \wt((\delta_a \otimes I_{n_b}) P_c^i) + \wt((I_{n_a}\otimes\delta_b)P_c^i).
\end{eqnarray}

The maximum syndrome weight along this path is less than $d_b E_a+E_b$. This follows because for the part $\wt((\delta_a \otimes I_{n_b}) P_c^i)$, each column contributes a syndrome weight of at most $E_a$, and there are $d_b$ columns where $L_b$ has support, which gives $d_bE_a$. While for the part $\wt((I_{n_a}\otimes\delta_b)P_c^i)$, the maximum syndrome weight is $E_b$ because we implement $L_b$ sequentially.

Analogously, for column-wise implementation, let $r_b = \{P_b^0, P_b^1, \ldots, P_b^F\}$ represent an optimal path for implementing $L_b$, with maximum syndrome weight $E_b$. One can construct another path for $L_c$ by implementing $L_a$ on columns according to the positions provided by path $r_b$. The maximum syndrome weight along this path is less than $d_a E_b+E_a$.

By selecting the implementation strategy with the lower maximum syndrome weight, we establish the upper bound:
\begin{equation}\label{eq:delta_upper}
\Delta(L_c) \leqslant \min [d_bE_a+E_b, d_aE_b+E_a].
\end{equation}

In addition, we assert that (although we cannot provide a formal proof) there is no path with an energy barrier lower than $\min [d_bE_a, d_aE_b]$, which we formally state in the following conjecture.
\begin{conjecture}
\label{mainconjt}
    Given two classical codes $\delta_a$ and $\delta_b$ with parameters $[n_a, k_a, d_a, E_a]$ and $[n_b, k_b, d_b, E_b]$, the energy barrier of the tensor product code $\delta_c$ is
    \begin{eqnarray}
        \Delta(\delta_c) \geqslant \min[d_aE_b, E_ad_b].
    \end{eqnarray}
\end{conjecture}

We conjecture this result and anticipate that an elegant, rigorous proof will be found in the future (See more discussion in the Appendix). Meanwhile, the lower bound provided in Lemma~\ref{lemma:energy_barrier_of_tensor_product_codes} is optimal with respect to the order of the distance, although the value $E_a$ and $E_b$ may yield a non-negligible difference in the absolute value of the energy barrier. Thus the lower bound in Lemma~\ref{lemma:energy_barrier_of_tensor_product_codes} is tight if underlying codes exhibit system size-dependent distances but constant energy barriers.

\section{Energy barrier of 3D hypergraph product codes}
\label{sec:energy_barrier_of_3d_hgp}

With all the necessary preparations in place, we now turn our attention to the energy barrier of HHGP codes. In this section, we consider the simplest instance, 3D hypergraph product codes that are constructed from three classical codes. 

We begin by reviewing the structure of 3D hypergraph product codes and analyzing their logical operators, defining a complete set as elementary canonical logical operators. We first address Z logical operators. Using the strip argument from Ref.~\cite{rakovszky2024physicsgoodldpccodes}, we establish an upper bound for the energy barrier. We then derive a lower bound by examining the structure of those elementary canonical Z logical operators. Based on energy barrier of tensor product codes (Lemma~\ref{lemma:energy_barrier_of_tensor_product_codes}), we demonstrate that this lower bound relates to the distances of the underlying classical codes. Finally, we discuss the energy barrier of the logical X operators.

\subsection{Code construction of 3D hypergraph product codes}

Consider three classical codes $\delta_a$, $\delta_b$, and $\delta_c$ with respective parameters $[n_{\ell}, k_{\ell}, d_{\ell}]$, where $\ell \in \{a,b,c\}$, and the parameters of their transposes $\delta_{\ell}^T$ are denoted by $[n_{\ell}^T, k_{\ell}^T, d_{\ell}^T]$. For convenience, we denote $r_\ell \equiv n_\ell^T$. Applying the 3D hypergraph product to these classical codes generates a length-4 chain complex~\cite{Quintavalle2021single}
\begin{equation}
    C_0 \xleftarrow{\partial_0} C_1 \xleftarrow{\partial_1} C_2 \xleftarrow{\partial_2} C_3,
\end{equation}
where $\partial_0, \partial_1$ and $\partial_2$ are
\begin{equation}
\partial_0=\left(\begin{array}{l}
\delta_a \otimes I_{n_b} \otimes I_{n_c} \\
I_{n_a} \otimes \delta_b \otimes I_{n_c} \\
I_{n_a} \otimes I_{n_b} \otimes \delta_c
\end{array}\right),
\label{eq:delta0}
\end{equation}

\begin{equation}
\partial_1=\left(\begin{array}{ccc}
I_{r_a} \otimes \delta_b \otimes I_{n_c}  & \delta_a \otimes I_{r_b} \otimes I_{n_c}  & 0_{r_ar_bn_c \times n_an_br_c} \\
I_{r_a} \otimes I_{n_b} \otimes \delta_c & 0_{r_an_br_c \times n_ar_bn_c} & \delta_a \otimes I_{n_b} \otimes I_{r_c}  \\
0_{n_ar_br_c \times r_an_bn_c} & I_{n_a} \otimes I_{r_b} \otimes \delta_c & I_{n_a} \otimes \delta_b \otimes I_{r_c} 
\end{array}\right),
\label{eq:delta1}
\end{equation}

\begin{equation}
\partial_2=\left(I_{r_a}  \otimes I_{r_b}  \otimes \delta_c \quad I_{r_a}  \otimes \delta_b \otimes I_{r_c}  \quad \delta_a \otimes I_{r_b}  \otimes I_{r_c} \right) .
\end{equation}
Here $I_{n}$ denotes an identity matrix of size $n \times n$. 

The relation $\partial_1\partial_0 = 0$ allows us to construct a quantum CSS code $\mathcal{C}\left(\delta_a, \delta_b, \delta_c\right)$ with parity check matrix $H_Z=\partial_0^T$ and $H_X=\partial_1$. The commutation relation follows as $H_XH_Z^T = \partial_1\partial_0 = 0$.

Additionally, since $\partial_2\partial_1 = 0$, the matrix $M =\partial_2$ serves as the meta-check for $H_X$. This means any valid X-syndrome must satisfy the constraints imposed by $M$.

The parameters of the code $\mathcal{C}\left(\delta_a, \delta_b, \delta_c\right)$ can be determined from the classical code $\delta_a$, $\delta_b$, and $\delta_c$. As demonstrated in Ref.~\cite{zeng2019higher}, the resulting quantum code $\mathcal{C}\left(\delta_a, \delta_b, \delta_c\right)$ is an $\left[\left[n, k, d_x, d_z\right]\right]$ code with
\begin{eqnarray}
\begin{aligned}
n & =n_a^T n_b n_c+n_a n_b^T n_c+n_a n_b n_c^T, \\
k & =k_a^T k_b k_c+k_a k_b^T k_c+k_a k_b k_c^T, \\
d_x & =\min \left[d_a^T, d_b^T, d_c^T\right], \\
d_z & =\min \left[d_b d_c, d_a d_c, d_a d_b\right] .
\end{aligned}
\end{eqnarray}

\subsection{Logical operators of 3D hypergraph product codes}

The logical operators of stabilizer codes commute with all elements in the stabilizer group while remaining distinct from it. We first examine the Z-logical operators, which are operators in $\ker(H_X) \backslash \img(H_Z^T)$.

Denote the binary form of logical operators of classical codes $\delta_a$, $\delta_b$, and $\delta_c$ as $L_A = \{L_A^1, L_A^2, \ldots, L_A^{k_a}\}$, $L_B = \{L_B^1, L_B^2, \ldots, L_B^{k_b}\}$, and $L_C = \{L_C^1, L_C^2, \ldots, L_C^{k_c}\}$, respectively. We first show that the following three sets of binary vectors give a complete set of Z logical operators.
\begin{equation}
\begin{aligned}
\label{3dzlogicaloperators}
    &L_Z^1 = \sum_{i,j,k}\alpha_{ijk}\left(\begin{array}{c}
         v_i\otimes L_B^j \otimes L_C^k \\
         0_{n_ar_bn_c} \\
         0_{n_an_br_c}
    \end{array}\right),\\
    &L_Z^2 = \sum_{i,j,k} \beta_{ijk}\left(\begin{array}{c}
    0_{r_an_bn_c}\\
    L_A^i \otimes w_j \otimes L_C^k\\
    0_{n_an_br_c}
    \end{array}\right),\\
    &L_Z^3 = \sum_{i,j,k}\gamma_{ijk} \left(\begin{array}{c}
    0_{r_an_bn_c}\\
    0_{n_ar_bn_c}\\
    L_A^i \otimes L_B^j \otimes u_k 
    \end{array}\right),\\
\end{aligned}
\end{equation}
where $v_i, w_j$, and $u_k$ are unit vectors with length $r_a, r_b$, and $r_c$, respectively. $0_{n}$ is a zero vector with length $n$, and $\alpha_{ijk}, \beta_{ijk}, \gamma_{ijk} \in \{0,1\}$. Given that $\delta_a L_A^i = \delta_b L_B^j = \delta_c L_C^k = 0$, one can confirm that $H_XL_Z^l = 0$ for $l \in \{1,2,3\}$. 

For $L_Z^1$, the condition $H_XL_Z^1 = 0$ ensures that any X-type stabilizer from $H_X$ commutes with $L_Z^1$. Similarly, any Z-type stabilizer from $H_Z$ naturally commutes with $L_Z^1$ as they both consist solely of Pauli Z operators. To prove that $L_Z^1$ are Z logical operators, we must prove that $L_Z^1$ are not in the image of $H_Z^T$. To proceed with this analysis, we first denote:
\begin{eqnarray}
    L_Z^{1,i,j,k} =\left(\begin{array}{c}
         v_i\otimes L_B^j \otimes L_C^k \\
         0_{n_ar_bn_c} \\
         0_{n_an_br_c}
    \end{array}\right).
    \label{eq:elementary_canonical_logical}
\end{eqnarray} 
For $a_i \in \{0,1\}$, suppose that the linear combination of these operators 
\begin{eqnarray}
    \sum_i a_iL_Z^{1,i,j,k} =\left(\begin{array}{c}
         (\sum_i a_iv_i) \otimes L_B^j \otimes L_C^k \\
         0_{n_ar_bn_c} \\
         0_{n_an_br_c}
         \end{array}\right)
\end{eqnarray}
is in the image of matrix $H_Z^T$. Then, there exists a vector $x$ such that
\begin{eqnarray}
   \sum_i a_i L_Z^{1,i,j,k} &=& \left(\begin{array}{c}
     (\sum_i a_iv_i) \otimes L_B^j \otimes L_C^k \\
     0_{n_ar_bn_c} \\
     0_{n_an_br_c}
   \end{array}\right) \nonumber \\
   &=& H_Z^T(x \otimes L_B^j \otimes L_C^k) \nonumber \\
   &=&  \left(\begin{array}{c}
     (\delta_a \otimes I_{n_b} \otimes I_{n_c})(x \otimes L_B^j \otimes L_C^k) \\
     (I_{n_a} \otimes \delta_b \otimes I_{n_c})(x \otimes L_B^j \otimes L_C^k) \\
     (I_{n_a} \otimes I_{n_b} \otimes \delta_c)(x \otimes L_B^j \otimes L_C^k)
   \end{array}\right) \nonumber \\
   &=& \left(\begin{array}{c}
     \delta_a x \otimes L_B^j \otimes L_C^k \\
     0_{n_ar_bn_c}\\
     0_{n_an_br_c}
   \end{array}\right).
\end{eqnarray}

Recall the number of encoded bits for $\delta_a^T$ is $k_a^T$, then $\dim(\img(\delta_a)) = n_a - k_a^T$, it is feasible to select a minimum of $k_a^T$ unit vectors $v_i$, ensuring that any combination of these $v_i$ remains external to the image of $\delta_a$. Specifically, for these $k_a^T$ unit vectors $v_i$, there does not exist any vector $x$ such that
\begin{eqnarray}
  \sum_{i=1}^{k_a^T} a_i v_i = \delta_a x.
\end{eqnarray}

For fixed $L_B^j$ and $L_C^k$, these $k_a^T$ unit vectors generate $k_a^T$ operators $L_Z^{1,i,j,k}$ that belong to $\ker(H_X) \backslash \img(H_Z^T)$. Since $\delta_b$ has $k_b$ logical bits and $\delta_c$ has $k_c$ logical bits, there are $k_b$ distinct $L_B$ and $k_c$ distinct $L_C$ operators. This yields $k_a^Tk_bk_c$ unique instances of $L_Z^{1, i,j,k}$, which function as logical Z operators of the quantum code.

Applying the same logic to operators $L_Z^2$ and $L_Z^3$ defined in Eq.~(\ref{3dzlogicaloperators}), we can demonstrate that there exist $k_ak_b^Tk_c$ distinct $L_Z^{2,i,j,k}$ and $k_ak_bk_c^T$ distinct $L_Z^{3,i,j,k}$ that serve as logical Z operators of the quantum code.

This completely characterizes all $k_a^Tk_bk_c + k_ak_b^Tk_c + k_ak_bk_c^T$ logical Z operators of the 3D hypergraph product codes. We refer to these as \textit{elementary canonical} logical Z operators.

We now examine the logical X operators, which are in the set $\ker(H_Z) \backslash \img(H_X^T)$. Consider the following three types of Pauli X operators,
\begin{eqnarray}
    &L_X^1 = \sum_{i,j,k} \alpha_{ijk}' \left(\begin{array}{c}
         L_A^{i'}\otimes w_j \otimes u_k  \\
         0_{n_ar_bn_c} \\
         0_{n_an_br_c}
    \end{array}\right),\\
    &L_X^2 = \sum_{i,j,k} \beta_{ijk}'\left(\begin{array}{c}
    0_{r_an_bn_c}\\
    v_i \otimes L_B^{j'} \otimes u_k\\
    0_{n_an_bn_c}
    \end{array}\right),\\
    &L_X^3 = \sum_{i,j,k} \gamma_{ijk}'\left(\begin{array}{c}
    0_{r_an_bn_c}\\
    0_{n_ar_bn_c}\\
    v_i \otimes w_j \otimes L_C^{k'}
    \end{array}\right).\\
\end{eqnarray}
Here $v_i, w_j$, and $u_k$ are unit vectors of length $n_a, n_b$ and $n_c$ respectively. Moreover, $L_A^{i'}, L_B^{j'}$ and $L_C^{k'}$ are logical operators of the code defined by the parity check matrix $\delta_a^T, \delta_b^T$, and $\delta_c^T$ respectively. 

For $L_X^1$, $H_ZL_X^1 = 0$ implies that any Z-type stabilizer from $H_Z$ commutes with $L_X^1$. Meanwhile, any X-type stabilizer in $H_X$ also commutes with $L_X^1$ as they consist only of Pauli X operators. To identify logical X operators, we need to find Pauli X operators in the set $L_X^1$ that are not present in the image of $H_X^T$. We denote:
\begin{eqnarray}
    L_X^{1,i,j,k} = \left(\begin{array}{c}
         L_A^{i'}\otimes w_j \otimes u_k  \\
         0_{n_ar_bn_c} \\
         0_{n_an_br_c}
    \end{array}\right).
\end{eqnarray}
For $b_j, c_k \in \{0,1\}$, suppose that the linear combination of these operators 
\begin{eqnarray}
    \sum_{jk} b_jc_k L_X^{1,i,j,k} = \left(\begin{array}{c}
         L_A^{i'}\otimes \sum_j b_j w_j \otimes \sum_k c_k u_k  \\
         0_{n_ar_bn_c} \\
         0_{n_an_br_c}
    \end{array}\right)
\end{eqnarray}
is in the image of matrix $H_X^T$. Then, there exists vectors $y, z$ such that
\begin{eqnarray}
   \sum_{jk} b_j c_k L_X^{1,i,j,k} &=& \left(\begin{array}{c}
     L_A^{i'}\otimes \sum_j b_j w_j \otimes \sum_k c_k u_k\\
     0_{n_ar_bn_c} \\
     0_{n_an_br_c}
   \end{array}\right) \nonumber \\
   &=& H_X^T(L_A^{i'} \otimes y \otimes z)
\end{eqnarray}

Recall that $k_b^T$ and $k_c^T$ represent the number of encoded bits for $\delta_b^T$ and $\delta_c^T$ respectively. Then with $\dim(\img(\delta_b^T)) = n_b - k_b$ and $\dim(\img(\delta_c^T)) = n_c - k_c$, one can select at least $k_b$ unit vectors $w_j$ and $k_c$ unit vectors $u_k$ such that any combination of these vectors remains outside the image of $\delta_b^T$ and $\delta_c^T$ respectively.

Given that there are $k_a^T$ different $L_A^{i'}$, and paralleling our analysis of logical Z operators, we can verify that there exist $k_a^Tk_bk_c$ distinct $L_X^{1,i,j,k}$, $k_ak_b^Tk_c$ distinct $L_X^{2,i,j,k}$, and $k_ak_bk_c^T$ distinct $L_X^{3,i,j,k}$ that serve as logical X operators in $\ker(H_Z) \backslash \img(H_X^T)$, we call those logical opertors as elementary canonical logical X operators. This accounts for all $k_a^Tk_bk_c + k_ak_b^Tk_c + k_ak_bk_c^T$ logical X operators of the code.

\subsection{Energy barrier of Z logical operators}

We now analyze the energy barrier for logical Z operators in 3D hypergraph product codes. First, we employ a strip argument to establish an upper bound. Then, using Lemma~\ref{lemma:energy_barrier_of_tensor_product_codes}, we derive a lower bound for the energy barrier for logical Z operators.

\subsubsection{Upper bound: ``strip'' argument}

The upper bound on the energy barrier can be obtained by considering a specific path from the identity to any logical operator. By choosing this logical operator to be an elementary canonical logical operator, and considering the ``strip'' argument mentioned in Sec.~\ref{sub:energy_barrier_tensor_product}, one can obtain an upper bound.

For elementary canonical logical operator in the form of $L_Z^{1,i,j,k}$, with fixed $i, j$, and $k$, the logical operator takes the form as in the Eq.~(\ref{eq:elementary_canonical_logical}). Since $v_i$ is a unit vector, the effective logical operator simplifies to $L_Z^{1^{\prime},j,k}=L_B^j \otimes L_C^k$. This can be interpreted exactly as a logical operator of a tensor product code. One can construct the same kind of strip configuration when implementing these logical operators. Thus, the energy barrier of this logical operator is upper bounded by the strip argument as $\Delta(L_Z^{1^{\prime},j,k}) \leqslant \min \left[E_b d_c, E_c d_b\right]$, where $E_b$ and $E_c$ are the energy barriers of the logical operators $L_B^j$ and $L_C^k$, respectively, and $d_b$ and $d_c$ are their corresponding distances.

With the same argument, the energy barrier of the canonical logical Z operator $L_Z^{2,i,j,k}$ is upper bounded by $\min \left[E_a d_c, E_c d_a\right]$, and the energy barrier of $L_Z^{3,i,j,k}$ is upper bounded by $\min \left[E_a d_b, E_b d_a\right]$.

In summary, consider a 3D hypergraph product code $\mathcal{C}$ constructed from classical codes $\delta_a$, $\delta_b$, and $\delta_c$ with parameters $\left[n_{\ell}, k_{\ell}, d_{\ell}, E_{\ell}\right]$, where $\ell \in \{a, b, c\}$. Based on the structure of the canonical logical operators we have derived, the energy barrier of the logical Z operators can be shown to have the following upper bound.
\begin{eqnarray}
    \Delta(L_Z) \leqslant  \min \{d_\ell E_m: \ell,m \in \{a,b,c\}, \ell \neq m \}.
\end{eqnarray}

\subsubsection{Lower bound}

We now establish a lower bound for the energy barrier of Z-type logical operators by analyzing the structure of these canonical logical operators. Consider a 3D hypergraph product code $\mathcal{C}$ constructed from classical codes $\delta_a$, $\delta_b$, and $\delta_c$ with parameters $\left[n_{\ell}, k_{\ell}, d_{\ell}, E_{\ell}\right]$, where $\ell \in \{a, b, c\}$. The parity check matrices of this quantum code are $H_Z=\partial_0^T$ and $H_X=\partial_1$, as defined in Eq.~(\ref{eq:delta0}) and Eq.~(\ref{eq:delta1}). Recall that the Z-type canonical logical operators $L_Z$ fall into three categories.
\begin{equation}
\begin{aligned}
    &L_Z^1 = \sum_{i,j,k} \alpha_{ijk}((v_i\otimes L_B^j \otimes L_C^k)^T, 0^T, 0^T)^T, \\
    &L_Z^2 = \sum_{i,j,k} \beta_{ijk}(0^T, (L_A^i \otimes w_j \otimes L_C^k)^T, 0^T)^T, \\
    &L_Z^3 = \sum_{i,j,k} \gamma_{ijk}(0^T, 0^T,  (L_A^i \otimes L_B^j \otimes u_k)^T )^T. 
\end{aligned}
\end{equation}
where (i) $\delta_a L_A^i = \delta_b L_B^j = \delta_c L_C^k = 0$ and (ii) $v_i \notin \operatorname{Im}\left(\delta_a\right)$, $w_j \notin \operatorname{Im}\left(\delta_b\right)$  and $u_k \notin \operatorname{Im}\left(\delta_c\right)$ are unit vectors. We note that $i \in \{1, \ldots, k_a^T\}$, $j \in \{1, \ldots, k_b^T\}$  and $k \in \{1, \ldots, k_c^T\} $ and that the set of logical operators expressible in this form is complete. A canonical logical operator is elementary if only one of the coefficients, $\alpha_{ijk}$ (also $\beta_{ijk}$ or $\gamma_{ijk}$), is equal to one.

The Tanner graph representation provides an intuitive framework that is well-suited for our analysis. Let $\mathcal{G}_1(V_1, C_1)$, $\mathcal{G}_2(V_2, C_2)$ and $\mathcal{G}_3(V_3, C_3)$ be the Tanner graphs of codes defined by $\delta_a$, $\delta_b$ and $\delta_c$, respectively. Here, $V_1$, $V_2$ and $V_3$ represent the set of bits, and $C_1$, $C_2$ and $C_3$ denote the set of checks. We use $\{v_1^i: i \in \{1,2,\cdots n_a\}\}$ (resp. $\{v_2^j: j \in \{1,2,\cdots n_b\}\}$, $\{v_3^k: j \in \{1,2,\cdots n_c\}\}$) to refer to bit vertices in $V_1$ (resp. $V_2$, $V_3$), and also as length $n_a$ (resp. $n_b$, $n_c$) unit vectors with the $i$th (resp. $j$th, $k$th) entry as $1$. 

The 3D hypergraph product $\mathcal{G}_1 \times \mathcal{G}_2 \times \mathcal{G}_3$ is a bipartite graph with vertex set $V \cup C$, where
\begin{eqnarray}
    V = V_1 \otimes V_2 \otimes C_3 \cup V_1 \otimes C_2 \otimes V_3 \cup C_1 \otimes V_2 \otimes V_3
\end{eqnarray} is the qubit set, and 
\begin{eqnarray}
    C &=& V_1 \otimes V_2 \otimes V_3 \cup C_1 \otimes C_2 \otimes V_3 \cup C_1 \otimes V_2 \otimes C_3 \nonumber\\
    &&\cup V_1 \otimes C_2 \otimes C_3 \cup C_1 \otimes C_2 \otimes C_3
\end{eqnarray}
is the check and meta-check set. More specifically, $V_1 \otimes V_2 \otimes V_3$ is the set the Z type checks, $C_1 \otimes C_2 \otimes V_3 \cup C_1 \otimes V_2 \otimes C_3 \cup V_1 \otimes C_2 \otimes C_3$ represents the set of X checks, and $C_1 \otimes C_2 \otimes C_3$ is meta-checks on X checks.

The set of qubits can be partitioned into three subsets, $V_1 \otimes V_2 \otimes C_3$, $V_1 \otimes C_2 \otimes V_3$, and $C_1 \otimes V_2 \otimes V_3$. For $H_X = \partial_1$, there are three column blocks (see Eq.~(\ref{eq:delta1})). The first column block of $H_X$ acts on $C_1 \otimes V_2 \otimes V_3$, the second column block of $H_X$ acts on $V_1 \otimes C_2 \otimes V_3$, and the third column block acts on $V_1 \otimes V_2 \otimes C_3$. Moreover, the subset $C_1 \otimes V_2 \otimes V_3$ can be further partitioned into $n_a^T \equiv r_a$ subsets $\{c_1^1 \otimes V_2 \otimes V_3, c_1^2 \otimes V_2 \otimes V_3, \cdots, c_1^{r_a} \otimes V_2 \otimes V_3 \}$, where $c_1^k \otimes V_2 \otimes V_3 := \{c_1^k \otimes y \otimes z: y\in V_2, z \in V_3 \}$ and $V_2 = \{v_2^1,\ldots, v_2^{n_b} \}$, $V_3 = \{v_3^1,\ldots, v_3^{n_c} \}$.

Thus, a Z-type Pauli operator can be expressed as a bit-string $z=\left({z^{(1)}}^T, {z^{(2)}}^T, {z^{(3)}}^T\right)^T$, where $z^{(1)}$ is supported on the qubit subset $C_1 \otimes V_2 \otimes V_3$ with vector space $\mathbb{F}_2^{r_a} \otimes \mathbb{F}_2^{n_b} \otimes \mathbb{F}_2^{n_c}$, $z^{(2)}$ is supported on the qubit set $V_1 \otimes C_2 \otimes V_3$ with vector space $\mathbb{F}_2^{n_a} \otimes \mathbb{F}_2^{r_b} \otimes \mathbb{F}_2^{n_c}$ and $z^{(3)}$ is supported on the qubit set $V_1 \otimes V_2 \otimes C_3$ with vector space $\mathbb{F}_2^{n_a} \otimes \mathbb{F}_2^{n_b} \otimes \mathbb{F}_2^{r_c}$. Because $z^{(1)}$, $z^{(2)}$, and $z^{(3)}$ act as a tensor product of three vector spaces, one can view them also as 3D matrices. (For instance, $v_i \otimes u_j \otimes w_k$ for unit vectors $v_i$, $u_j$ and $w_k$ can be viewed as a 3D matrix whose entry is $1$ on the $i$'th a-slice, the $j$'th b-slice and the $k$'th c-slice and zero elsewhere. We call this procedure vector reshaping. After reshaping, $z^{(1)}$, $z^{(2)}$ and $z^{(3)}$ become $Z^{(1)}$, $Z^{(2)}$ and $Z^{(3)}$ respectively. Here, $Z^{(1)}$ is an $r_a \times n_b \times n_c$ matrix, and the entry $Z^{(1)}_{i,j,k}$ is supported on the qubit $r_1^i \otimes v_2^j \otimes v_3^k$. Moreover, the $i$th a-slice $Z^{(1)}_{i}$ is supported on qubits $\{r_1^i \otimes v_2^1\otimes v_3^1, r_1^i \otimes v_2^2\otimes v_3^1, \cdots, r_1^i \otimes v_2^{n_b}\otimes v_3^{n_c}\}$, which we refer to as $r_1^i \otimes V_2 \otimes V_3$. Similarly, the subset $V_1 \otimes C_2 \otimes V_3$ can be partitioned into $r_b$ b-slices, and $j$th b-slice is supported on the qubit subset $V_1 \otimes r_2^j \otimes V_3$.  Also, the subset $V_1 \otimes V_2 \otimes C_3$ can be partitioned into $r_c$ c-slices, and $k$th c-slice is supported on the qubit subset $V_1 \otimes V_2 \otimes r_3^k$ (defined similarly). 

An elementary canonical logical-Z operator, for example, in the form $(v_i\otimes L_B^j \otimes L_C^k, 0, 0)^T$, is supported on the subset $r_1^i \otimes V_2 \otimes V_3$ if $v_i=r_1^i$. In matrix form, it is supported on the $i$th a-slice of $Z^{(1)}$. Fig.~\ref{fg:3d_slices} provides an illustrative example of the 3D Toric code structure.

\begin{figure}[htbp]
  \centering
  \includegraphics[width=0.48\textwidth]{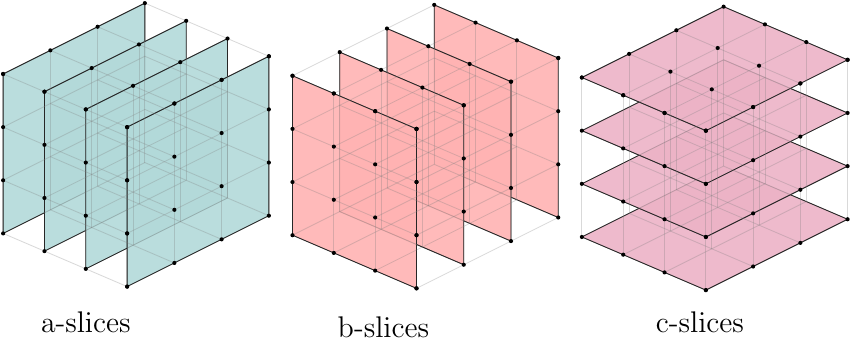}
  \caption{The 3D toric code can be structured as a cubic lattice (with periodic boundary conditions). From left to right, the figure illustrates its decomposition into a-slices, b-slices, and c-slices. The 3D Toric code encodes 3 logical qubits, with corresponding 3 logical Z operators. Each Z elementary canonical logical operator is supported on a single slice, with slices of the same color representing logical operators equivalent up to stabilizers. Generally, given any 3D hypergraph product code, one can create such a decomposition, although the structure within each slice will be more complicated and depends on the underlying classical codes.}
  \label{fg:3d_slices}
\end{figure}

We now prove a lower bound for the energy barrier of the canonical logical-Z operators, as shown in the following proposition.

\begin{proposition}\label{proposition:canonical_z_energy_barrier_3d}
    For any nontrivial canonical logical-Z operator $L$,
    \begin{equation}
        \Delta(L) \geq \min[d_a, d_b, d_c]
    \end{equation}
\end{proposition}

We introduce two lemmas to prove this proposition. We use the following convention in the proof. For the subsequent proof, we adopt the following notation: a path $r = \{P_0, P_1, \ldots, P_\ell\}$ is said to be \emph{supported on} a subset $U \subseteq V$ if $\mathrm{supp}(P_i) \subseteq U$ for all $0 \leq i \leq \ell$.

\begin{lemma}
\label{lemma:path_within_subset_3d}
    For any elementary canonical logical-Z operator $L$ supported on $r_1^\alpha \otimes V_2 \otimes V_3$ (resp. $V_1 \otimes r_2^\beta \otimes V_3$, $V_1 \otimes V_2 \otimes r_3^\gamma$), its energy barrier is attained by a path supported on $r_1^\alpha \otimes V_2 \otimes V_3$ (resp. $V_1 \otimes r_2^\beta \otimes V_3$, $V_1 \otimes V_2 \otimes r_3^\gamma$).
\end{lemma}

\begin{lemma}
\label{lemma:energy_barrier_of_composite_logical_3d}
    For any nontrivial canonical logical-Z operator $L$, the energy barrier $\Delta(L)$ is greater than or equal to the minimum energy barrier of the elementary canonical logical-Z operators.
\end{lemma}

Consider an elementary canonical logical-Z operator $L$ supported on $r_1^\alpha \otimes V_2 \otimes V_3$. Suppose $\Delta(L)$ is given by a path $r$. The main idea behind the proof of Lemma~\ref{lemma:path_within_subset_3d} is to deform a general path $r$ into a new path $r'$, supported only on $r_1^\alpha \otimes V_2 \otimes V_3$, and the energy barrier of $r'$ is not greater than that of the original path $r$. A similar argument can be applied to prove Lemma~\ref{lemma:energy_barrier_of_composite_logical_3d}. The proofs closely extended from the arguments in the hypergraph product case~\cite{zhao2024on}. The details are provided in the appendix [see Sec~\ref{sec:proofoflemma14} and ~\ref{sec:proofoflemma15}].

For the elementary canonical logical operator $L_Z^{1,\alpha} = ((r_1^\alpha \otimes L_B^j \otimes L_C^k)^T, 0^T, 0^T)^T$, with fixed $\alpha$, $j$, and $k$. $L_Z^{1,\alpha}$ is supported on $r_1^\alpha \otimes V_2 \otimes V_3$. To compute the energy barrier of $L_Z^{1,\alpha}$, consider a path $r = \{P_0, P_1, \cdots, P_F\}$. Then
\begin{equation}
    \Delta(L_Z^{1,\alpha}) = \min \{\epsilon_{\max}(r): r \in w(0, L_Z^{1,\alpha})\}.
\end{equation}
Here, $P_0 = (0,0,0)^T$ is a zero vector, and $P_F = ((r_1^\alpha \otimes L_B^j \otimes L_C^k)^T, 0^T, 0^T)^T$. Note that $r$ is a path that walks throughout the set of qubits. Therefore, $P_\ell \in r$ could be represented as $P_\ell = ((P_\ell^1)^T, (P_\ell^2)^T, (P_\ell^3)^T)^T$, where each component, $P_\ell^1, P_\ell^2$, and $P_\ell^3$, could be nontrivial.

Lemma~\ref{lemma:path_within_subset_3d} implies that there is a path $r$ supported on $r_1^\alpha \otimes V_2 \otimes V_3$ that attains the energy barrier of $\Delta(L)$. A path restricted to the first component means that for any $P_\ell \in r$, one can set $P_\ell^2 = P_\ell^3 = 0$, the only nontrivial Paulis (nonzero binaries) are in $P_\ell^1$ with fixed $r_1^\alpha$. Furthermore, $r_1^\alpha$ is fixed means every $P_\ell^1$ has the form $r_1^\alpha \otimes y \otimes z$, where $y \in V_2$ and $z \in V_3$ are vectors of length $n_b$ and $n_c$ respectively. 

Given a $P_\ell \in r$, the energy (the number of violated checks) of $P_\ell$ is given by $\epsilon(P_\ell) = \mathrm{wt}\left(H_XP_\ell\right)$, has the form
\begin{eqnarray}
    \mathrm{wt}\left(\left(
    {\setlength{\arraycolsep}{0.5pt}
    \begin{array}{@{}c@{\,}c@{\,}c@{}}
    I_{n_a} \!\otimes\! \delta_b \!\otimes\! I_{n_c}  & \delta_a \!\otimes\! I_{n_b} \!\otimes\! I_{n_c}  & 0 \\
    I_{n_a} \!\otimes\! I_{n_b} \!\otimes\! \delta_c & 0 & \delta_a \!\otimes\! I_{n_b} \!\otimes\! I_{n_c}  \\
    0 & I_{n_a} \!\otimes\! I_{n_b} \!\otimes\! \delta_c & I_{n_a} \!\otimes\! \delta_b \!\otimes\! I_{n_c} 
    \end{array}}
    \right) 
    {\setlength{\arraycolsep}{1pt}
    \left(\begin{array}{@{}c@{}}
     P_\ell^1 \\
     P_\ell^2 \\
     P_\ell^3
    \end{array}\right)}\right). \nonumber \\
\end{eqnarray}

By setting $P_\ell^2 = P_\ell^3 = 0$, then
\begin{equation}
    \begin{aligned}
        \epsilon(P_\ell) &= \mathrm{wt} \left(\left(\begin{array}{c}
          I_{n_a} \otimes \delta_b \otimes I_{n_c} \\
          I_{n_a} \otimes I_{n_b} \otimes \delta_c
        \end{array}
        \right) P_\ell^1
        \right) \\
        & = \mathrm{wt}\left(I_{n_a} \otimes \delta_b \otimes I_{n_c} P_\ell^1\right) + \mathrm{wt}\left(I_{n_a} \otimes I_{n_b} \otimes \delta_c P_\ell^1\right).
    \end{aligned}
\end{equation}

Furthermore, for fixed $r_1^\alpha$, because $r_1^\alpha$ is a unit vector, the part $r_1^\alpha$ can be discarded by applying the map $(r_1^\alpha)^T \otimes I_{n_b} \otimes I_{n_c}$ to $P_\ell^1$. Thus
\begin{equation}
    \begin{aligned}
        \epsilon(P_\ell) = \mathrm{wt}\left(\delta_b \otimes I_{n_c} P_\ell^{1'}\right) + \mathrm{wt}\left( I_{n_b} \otimes \delta_c P_\ell^{1'}\right),
    \end{aligned}
\end{equation}
with $P_\ell^{1'} = ((r_1^\alpha)^T \otimes I_{n_b} \otimes I_{n_c}) P_\ell^1$.

In summary, to calculate the energy barrier of the Z-type logical operator $L_Z^{1,\alpha,j,k} = ((r_1^\alpha \otimes L_B^j \otimes L_C^k)^T, 0^T, 0^T)^T$, one can consider a path $r \in \{Q_0, Q_1, \cdots, Q_L \}$ supported on qubit subset $r_1^\alpha \otimes V_2 \otimes V_3$, then
\begin{equation}
    \Delta(L_Z^{1,\alpha,j,k}) = \min \{ \epsilon_{\max}(r): r \in w(0, L_B^j\otimes L_C^k)\}.
\end{equation}
where for each $r$,
\begin{equation}
    \epsilon_{\max}(r) = \max[\epsilon(Q_0), \epsilon(Q_1), \cdots, \epsilon(Q_L)],
\end{equation}
and the energy cost of the step $Q_\ell$ is
\begin{equation}
\label{eq:energy_step}
    \epsilon(Q_\ell) = \mathrm{wt}\left(\delta_b \otimes I_{n_c} Q_\ell\right) + \mathrm{wt}\left( I_{n_b} \otimes \delta_c Q_\ell\right).
\end{equation}

This scenario corresponds exactly to computing the energy barrier of a classical tensor product code (Eq.~(\ref{eq:energy_pell_main})), where the underlying classical codes are $\delta_b$ and $\delta_c$. The subset $r_1^\alpha \otimes V_2 \otimes V_3$ is the bit set of this classical tensor product code. The path for applying logical operator $L_Z^{1, \alpha,j,k}$ effectively forms a path for the logical operator $L_B \otimes L_C$ in the tensor product code constructed by $\delta_b$ and $\delta_c$.

According to Lemma~\ref{lemma:energy_barrier_of_tensor_product_codes}, the energy barrier of $L_B \otimes L_C$ is lower bounded by $\min \left[d_b, d_c\right]$. Thus the energy barrier of the Z-type logical operator $L_Z^{1, \alpha,j,k}$ satisfies $\Delta\left(L_Z^{1, \alpha,j,k}\right) \geq$ $\min \left[d_b, d_c\right]$.

Following the same reasoning, the energy barrier for logical Z-type operators of the form $L_Z^{2,i,j,k}=$ $\left(0^T, (L_A^i \otimes w_j \otimes L_C^k)^T, 0^T\right)^T$ is lower bounded by $\min \left[d_a, d_c\right]$, and the logical operators of the form $L_Z^{3,i,j,k}=\left(0^T,0^T, (L_A^i \otimes L_B^j \otimes u_k)^T\right)^T$ have an energy barrier lower bounded by $\min \left[d_a, d_b\right]$.

In summary, for any elementary canonical Z logical operator $L$ in the 3D hypergraph product code, the energy barrier is lower bounded by
\begin{eqnarray}
    \Delta(L) \geqslant \min \left[d_a, d_b, d_c\right].
\end{eqnarray}

With Lemma \ref{lemma:energy_barrier_of_composite_logical_3d}, we can establish that for any canonical Z-type logical operator (not just elementary ones), the energy barrier is lower bounded by $\min \left[d_a, d_b, d_c\right]$. Moreover, results from Ref.~\cite{zhao2024on} demonstrate that for any logical operator $L$ and stabilizer $S$ with $\Delta(L) \geqslant \Delta(S)$, we have $\Delta(L) = \Delta(LS)$. Therefore, with the LDPC property yielding $\Delta(S) = O(1)$ for any $S$, and in cases where $d_a, d_b, d_c \geqslant O(1)$, this lower bound applies to all logical Z operators, not just the canonical ones.

Formally, we have the following theorem.
\begin{thm}
\label{thm:energy_barrier_3d_hgp}
Let $\Delta(L_Z)$ be the energy barrier of the Z logical operators of the 3D hypergraph product code constructed from classical codes $\delta_a$, $\delta_b$ and $\delta_c$ with parameters $[n_\ell, k_\ell, d_\ell, E_\ell]$ for $\ell \in \{a,b,c\}$. If the resulting quantum code is LDPC with sparsity parameters $w_c, w_q$, then in cases where $d_a, d_b, d_c \geqslant w_cw_q$, we have
\begin{eqnarray}
	\Delta(L_Z) \geqslant \min [d_a, d_b, d_c].
\end{eqnarray}
\end{thm}

Similarly, with Lemma~\ref{lemma:energy_barrier_tensor_product_energy}, we can conclude that the energy barrier of 3D LDPC hypergraph product codes is lower bounded by the energy barriers of the underlying classical codes:
\begin{eqnarray}
    \Delta(L_Z) \geqslant \min [E_a, E_b, E_c].
\end{eqnarray}

Moreover, if Conjecture~\ref{mainconjt} is true, we have
\begin{eqnarray}
    \Delta(L_Z) \geqslant \min \{d_\ell E_m: \ell,m \in \{a,b,c\}, \ell \neq m \}.
\end{eqnarray}
This lower bound matches the upper bound obtained using the strip argument, effectively providing a tight bound.

\subsection{Energy barrier of X logical operators}

The canonical X logical operators have the form
\begin{eqnarray}
    &L_X^1 = \sum_{i,j,k} \alpha_{ijk}'((L_A^{i'}\otimes w_j \otimes u_k)^T, 0_{n_ar_bn_c}^T, 0_{n_an_br_c}^T)^T, \nonumber \\
    &L_X^2 = \sum_{i,j,k} \beta_{ijk}'(0_{r_an_bn_c}^T, (v_i \otimes L_B^{j'} \otimes u_k)^T, 0_{n_an_bn_c}^T)^T, \nonumber \\
    &L_X^3 = \sum_{i,j,k} \gamma_{ijk}'(0_{r_an_bn_c}^T, 0_{n_ar_bn_c}^T,  (v_i \otimes w_j \otimes L_C^{k'})^T )^T, \nonumber \\
\end{eqnarray}
where (i) $\delta_a^T L_A^{i'} = \delta_b^T L_B^{j'} = \delta_c^T L_C^{k'} = 0$ and (ii) $v_i \notin \operatorname{Im}\left(\delta_a^T\right)$, $w_j \notin \operatorname{Im}\left(\delta_b^T\right)$  and $u_k \notin \operatorname{Im}\left(\delta_c^T\right)$ are unit vectors. We note that $i \in \{1, \ldots, k_a^T\}$, $j \in \{1, \ldots, k_b\}$  and $k \in \{1, \ldots, k_c\} $ and that the set of logical operators expressible in this form is complete. A canonical logical operator is elementary if only one of the coefficients, $\alpha_{ijk}'$, $\beta_{ijk}'$ or $\gamma_{ijk}'$, is equal to one.

These elementary canonical logical X operators take the form of tensor products between one classical logical operator and two unit vectors, which resembles that of logical operators in hypergraph product codes, though in a higher dimension. By extending the reasoning previously applied to hypergraph product codes~\cite{zhao2024on} to this higher-dimensional case, in the LDPC regime, one can establish a lower bound for the energy barrier of logical X operators as
\begin{eqnarray}
   \Delta(L_X) \geqslant  \min \left[E_a^T, E_b^T, E_c^T\right].
\end{eqnarray}

The upper bound can be also derived from the structure of these elementary canonical logical X operators~\cite{zhao2024on}, which is also $\min [E_a^T, E_b^T, E_c^T]$. In the case where $\min \left[E_a^T, E_b^T, E_c^T\right] \geqslant w_cw_q$, with $w_c, w_q$ are the sparsity parameters, we have
\begin{eqnarray}
   \Delta(L_X) =  \min \left[E_a^T, E_b^T, E_c^T\right].
\end{eqnarray}

\section{Energy barrier of 4D hypergraph product codes}
\label{sec:energy_barrier_of_4d_hgp}

In 3D hypergraph product codes, the energy barrier of Z-type logical operators is lower bounded by the distance of the underlying classical codes. While X logical operators retain the energy barrier of these classical codes. When classical codes possess large distances but small energy barriers, only Z-type logical operators have improved energy barrier, potentially enabling single-shot. The X-type logical operators, lacking this enhanced protection, remain vulnerable. In this section, we show that 4D hypergraph product codes offer improved energy barrier for both X and Z logical operators.

The distinction between the 3D and 4D cases stems from structural differences in their elementary canonical logical operators. In 3D codes, Z logical operators take the form $L_Z=L_{c_1} \otimes L_{c_2} \otimes u_k$, while X logical operators are simply $L_X= v_i\otimes w_j \otimes L_{c_3}^T$, where $v_i, w_j$, and $u_k$ are unit vectors, and $L_{c_1}, L_{c_2}, L_{c_3}$ are classical logical operators. In contrast, 4D codes exhibit a more balanced structure where both X and Z logical operators are tensor products of two classical logical operators (and two unit vectors, see Eq.~\ref{eq:4dhgp_lz} and Eq.~\ref{eq:4dhgp_lx}), resulting in enhanced energy barriers for both types of logical operators.

Given four classical codes represented by parity check matrices $\delta_a, \delta_b, \delta_c, \delta_d$, let $\left[n_{\ell}, k_{\ell}, d_{\ell}, E_{\ell}\right]$ and $\left[n_{\ell}^T, k_{\ell}^T, d_{\ell}^T, E_{\ell}^T\right]$ represent the parameters of the classical linear code defined by the parity check matrix $\delta_{\ell}$ and its transpose $\delta_{\ell}^T$. For convenience, we denote $r_\ell \equiv n_\ell^T$, where $\ell=\{a,b,c,d\}$. The 4D hypergraph product is defined by a length-$5$ chain complex
\begin{equation}
    C_0 \xleftarrow{\partial_0} C_1 \xleftarrow{\partial_1} C_2 \xleftarrow{\partial_2} C_3 \xleftarrow{\partial_3} C_4,
\end{equation}

Boundary maps $\partial_0, \partial_1$, $\partial_2$ and $\partial_3$ are
\onecolumngrid
\begin{eqnarray}
    \partial_0=\left(\begin{array}{c}\delta_a \otimes I_{n_b} \otimes I_{n_c} \otimes I_{n_d} \\ I_{n_a} \otimes \delta_b \otimes I_{n_c} \otimes I_{n_d} \\ I_{n_a} \otimes I_{n_b} \otimes \delta_c \otimes I_{n_d} \\ I_{n_a} \otimes I_{n_b} \otimes I_{n_c} \otimes \delta_d\end{array}\right),
\end{eqnarray}

\begin{eqnarray}
    \partial_1= \left(\begin{array}{cccc}
    I_{r_a} \otimes \delta_b \otimes I_{n_c} \otimes I_{n_d} & \delta_a \otimes I_{r_b} \otimes I_{n_c} \otimes I_{n_d} & 0 & 0 \\ 
    I_{r_a} \otimes I_{n_b} \otimes \delta_c \otimes I_{n_d} & 0 & \delta_a \otimes I_{n_b} \otimes I_{r_c} \otimes I_{n_d} & 0 \\ 
    I_{r_a} \otimes I_{n_b} \otimes I_{n_c} \otimes \delta_d & 0 & 0 & \delta_a \otimes I_{n_b} \otimes I_{n_c} \otimes I_{r_d} \\ 
    0 & I_{n_a} \otimes I_{r_b} \otimes \delta_c \otimes I_{n_d} & I_{n_a} \otimes \delta_b \otimes I_{r_c} \otimes I_{n_d} & 0 \\ 
    0 & I_{n_a} \otimes I_{r_b} \otimes I_{n_c} \otimes \delta_d & 0 & I_{n_a} \otimes \delta_b \otimes I_{n_c} \otimes I_{r_d} \\ 
    0 & 0 & I_{n_a} \otimes I_{n_b} \otimes I_{r_c} \otimes \delta_d & I_{n_a} \otimes I_{n_b} \otimes \delta_c \otimes I_{r_d}\end{array}\right), \nonumber \\
\end{eqnarray}

\footnotesize{
\begin{eqnarray}
    &&\partial_2= \nonumber\\
    &&\left({\setlength{\arraycolsep}{0.5pt}\begin{array}{cccccc}
    I_{r_a} \otimes I_{r_b} \otimes \delta_c \otimes I_{n_d} & I_{r_a} \otimes \delta_b \otimes I_{r_c} \otimes I_{n_d} & 0 & \delta_a \otimes I_{r_b} \otimes I_{r_c} \otimes I_{n_d} & 0 & 0 \\ 
    I_{r_a} \otimes I_{r_b} \otimes I_{n_c} \otimes \delta_d & 0 & I_{r_a} \otimes \delta_b \otimes I_{n_c} \otimes I_{r_d} & 0 & \delta_a \otimes I_{r_b} \otimes I_{n_c} \otimes I_{r_d} & 0 \\ 
    0 & I_{r_a} \otimes I_{n_b} \otimes I_{r_c} \otimes \delta_d & I_{r_a} \otimes I_{n_b} \otimes \delta_c \otimes I_{r_d} & 0 & 0 & \delta_a \otimes I_{n_b} \otimes I_{r_c} \otimes I_{r_d} \\ 
    0 & 0 & 0 & I_{n_a} \otimes I_{r_b} \otimes I_{r_c} \otimes \delta_d & I_{n_a} \otimes I_{r_b} \otimes \delta_c \otimes I_{r_d} & I_{n_a} \otimes \delta_b \otimes I_{r_c} \otimes I_{r_d}\end{array}}\right), \nonumber \\
\end{eqnarray}}

\normalsize
\begin{eqnarray}
    \partial_3=\left(\begin{array}{cccc} I_{r_a} \otimes I_{r_b} \otimes I_{r_c} \otimes \delta_d & I_{r_a} \otimes I_{r_b} \otimes \delta_c \otimes I_{r_d} & I_{r_a} \otimes \delta_b \otimes I_{r_c} \otimes I_{r_d} & \delta_a \otimes I_{r_b} \otimes I_{r_c} \otimes I_{r_d} \end{array}\right). \nonumber \\
\end{eqnarray}

\twocolumngrid

One can confirm that $\partial_i\partial_{i-1} = 0 \mod 2$ for $i \in {1,2,3}$. In this case, $H_Z=\partial_1^T, H_X=\partial_2$ is used to construct a quantum CSS code $H_{(\delta_a, \delta_b, \delta_c, \delta_d)}$ with parity check matrix
\begin{equation}
H_{\left(\delta_a, \delta_b, \delta_c, \delta_d\right)}=\left(\begin{array}{cc}
H_X & 0 \\
0 & H_Z
\end{array}\right).
\end{equation}
Moreover, as $\partial_3\partial_2 = \partial_1\partial_0 = 0$, the remaining matrices $M_1 =\partial_3$, $M_2= \partial_0^T$ work as a meta-check for $H_X$ and $H_Z$, respectively. This means that any valid X-syndrome adheres to the constraints defined by $M_1$, and any valid Z-syndrome satisfies the constraints defined by $M_2$.

The parameters of the resulting quantum code $\mathcal{C}\left(\delta_a, \delta_b, \delta_c, \delta_d\right)$ depend on the parameters of classical codes $\delta_a, \delta_b$, $\delta_c$, and $\delta_d$, where

\onecolumngrid
\begin{eqnarray}
n &=&n_a^T n_b^T n_c n_d+n_a^T n_b n_c^Tn_d+n_a^T n_b n_cn_d^T+n_a n_b^T n_c^Tn_d +n_a n_b^T n_cn_d^T +n_a n_b n_c^Tn_d^T, \nonumber\\
k &=& k_a^T k_b^T k_ck_d+ k_a^T k_b k_c^Tk_d+k_a^T k_b k_ck_d^T +k_a k_b^T k_c^Tk_d+k_a k_b^T k_ck_d^T+k_a k_b k_c^Tk_d^T, \nonumber \\
d_x &=& \min \left[d_a^Td_b^T, d_c^Td_a^T, d_a^Td_d^T, d_c^Td_b^T, d_b^Td_d^T, d_c^Td_d^T\right], \nonumber \\
d_z &=& \min \left[d_b d_c, d_b d_d, d_a d_d,d_ad_c,d_ad_b\right] .
\end{eqnarray}

The Z-type canonical logical operators $L_Z$ can be categorized into the following six sets.
\begin{equation}
\begin{aligned}
    &L_Z^1 = \sum_{i,j,k,l} a^1_{ijkl}((v_i\otimes w_j \otimes L_C^k \otimes L_D^l)^T, 0^T, 0^T, 0^T,0^T,0^T)^T, \\
    &L_Z^2 = \sum_{i,j,k,l} a^2_{ijkl}(0^T, (v_i \otimes L_B^j \otimes u_k \otimes L_D^l)^T, 0^T,0^T,0^T,0^T)^T, \\
    &L_Z^3 = \sum_{i,j,k,l} a^3_{ijkl}(0^T, 0^T,  (v_i \otimes L_B^j \otimes L_C^k \otimes t_l)^T,0^T,0^T,0^T )^T, \\
    &L_Z^4 = \sum_{i,j,k,l} a^4_{ijkl}(0^T, 0^T, 0^T,  (L_A^i \otimes w_j \otimes u_k \otimes L_D^l)^T,0^T,0^T )^T, \\
    &L_Z^5 = \sum_{i,j,k,l} a^5_{ijkl}(0^T, 0^T, 0^T,0^T, (L_A^i \otimes w_j \otimes L_C^k \otimes t_l)^T,0^T )^T, \\
    &L_Z^6 = \sum_{i,j,k,l} a^6_{ijkl}(0^T, 0^T,0^T,0^T,0^T, (L_A^i \otimes L_B^j \otimes u_k \otimes t_l)^T )^T,
\end{aligned}
\label{eq:4dhgp_lz}
\end{equation}

\twocolumngrid
where $v_i, w_j$, $u_k$, and $t_l$ are unit vectors with length $r_a, r_b$, $r_c$ and $r_d$, respectively. ${L_A^i}, {L_B^j}$, ${L_C^k}$ and ${L_D^l}$ are logical operators of the code defined by the parity check matrix $\delta_a, \delta_b$, $\delta_c$ and $\delta_d$, respectively. Given that $\delta_a L_A^i = \delta_b L_B^j = \delta_c L_C^k = \delta_d L_D^l = 0$, one can confirm that $H_XL_Z^l = 0$ for $l \in \{1,2,3,4,5,6\}$.

Similarly, the X-type canonical logical operators $L_X$ can be categorized into six sets.
\onecolumngrid
\begin{equation}
\begin{aligned}
    &L_X^1 = \sum_{i,j,k,l} b^1_{ijkl}(({L_A^i}' \otimes {L_B^j}' \otimes u_k \otimes t_l)^T,0^T,0^T,0^T,0^T,0^T )^T, \\
    &L_X^2 = \sum_{i,j,k,l} b^2_{ijkl}(0^T, ({L_A^i}' \otimes w^j \otimes {L_C^k}' \otimes t_l)^T,0^T ,0^T,0^T,0^T)^T, \\
    &L_X^3 = \sum_{i,j,k,l} b^3_{ijkl}(0^T, 0^T, ( {L_A^i}' \otimes w_j \otimes u_k \otimes {L_D^l}')^T,0^T,0^T,0^T )^T, \\
    &L_X^4 = \sum_{i,j,k,l} b^4_{ijkl}(0^T, 0^T, 0^T,  (v_i \otimes {L_B^j}' \otimes {L_C^k}' \otimes t_l)^T,0^T,0^T )^T, \\
    &L_X^5 = \sum_{i,j,k,l} b^5_{ijkl}(0^T, 0^T,0^T,0^T, (v_i \otimes {L_B^j}' \otimes u_k \otimes {L_D^l}')^T, 0^T)^T, \\
    &L_X^6 = \sum_{i,j,k,l} b^6_{ijkl}(0^T,0^T,0^T,0^T,0^T,(v_i\otimes w_j \otimes {L_C^k}' \otimes {L_D^l}')^T)^T.
\end{aligned}
\label{eq:4dhgp_lx}
\end{equation}
\twocolumngrid
where $v_i, w_j$, $u_k$ and $t_l$ are unit vectors of length $n_a, n_b$, $n_c$ and $n_d$, respectively. ${L_A^i}', {L_B^j}'$, ${L_C^k}'$ and ${L_D^l}'$ are logical operators of the code defined by the parity check matrix $\delta_a^T, \delta_b^T$, $\delta_c^T$ and $\delta_d^T$, respectively. 

Both X and Z elementary canonical logical operators have the structure of tensor product of two classical logical operators (and two unit vectors). Thus, in the LDPC regime, with Lemma~\ref{lemma:energy_barrier_of_tensor_product_codes}, by applying the same reasoning used for logical Z operators in 3D hypergraph product codes but extended to four dimension, we can establish the following lower bounds on the energy barriers for logical Z and X operators in the 4D hypergraph product code.
\begin{eqnarray}
    \Delta(L_Z) &\geqslant& \min [d_a, d_b, d_c, d_d], \nonumber\\
    \Delta(L_X) &\geqslant& \min [d_a^T, d_b^T, d_c^T, d_d^T].
    \label{eq:bound_4d}
\end{eqnarray}

We state it formally in the following Theorem
\begin{thm}
\label{thm:energy_barrier_4d_hgp}
Let $\Delta(L_Z)$, $\Delta(L_X)$ be the energy barrier of the Z, X logical operators of the 4D hypergraph product code constructed from classical codes $\delta_a$, $\delta_b$, $\delta_c$ and $\delta_d$, with parameters $[n_\ell, k_\ell, d_\ell, E_\ell]$ for $\ell \in \{a,b,c, d\}$. If the resulting quantum code is LDPC with sparsity parameters $w_c, w_q$. When $d_\ell, d_\ell^T \geqslant w_cw_q$ for $\ell \in \{a,b,c,d\}$,  then we have
\begin{eqnarray}
	\Delta(L_Z) &\geqslant& \min [d_a, d_b, d_c, d_d], \\
    \Delta(L_X) &\geqslant& \min [d_a^T, d_b^T, d_c^T, d_d^T].
\end{eqnarray}   
\end{thm}

Similarly, with Lemma~\ref{lemma:energy_barrier_tensor_product_energy}, we have
\begin{eqnarray}
	\Delta(L_Z) &\geqslant& \min [E_a, E_b, E_c, E_d], \\
    \Delta(L_X) &\geqslant& \min [E_a^T, E_b^T, E_c^T, E_d^T].
\end{eqnarray}

Moreover, if Conjecture~\ref{mainconjt} is true, we have
\begin{eqnarray}
    \Delta(L_Z) &\geqslant& \min \{d_\ell E_m: \ell,m \in \{a,b,c,d\}, \ell \neq m \}, \\
    \Delta(L_X) &\geqslant& \min \{ d_\ell^TE_m^T: \ell,m \in \{a,b,c,d\}, \ell \neq m \},
\end{eqnarray}
which is tight.

\section{Conclusion and outlook}

In this work, we first gave a introduction to the confinement, soundness, and expansion properties of codes. We also discussed the relationships among these properties. Specifically, we showed the equivalence between left-expansion and linear confinement, and demonstrated formally that confinement property inherently induces a lower bound on energy barrier. 

We then established a lower bound for the energy barrier of tensor product codes in Lemma~\ref{lemma:energy_barrier_of_tensor_product_codes}, showing its relation to the distance of the underlying codes. Using this lemma, we proved that for 3D LDPC hypergraph product codes constructed from classical codes $\delta_a$, $\delta_b$, and $\delta_c$ with parameters $\left[n_{\ell}, k_{\ell}, d_{\ell}, E_{\ell}\right]$, where $\ell \in {a, b, c}$, the energy barrier of logical Z operators is lower-bounded by $\Omega(d_\ell)$, which improves upon the $\Omega(d_\ell^{\frac13})$ bound derived from the soundness and confinement property. For logical X operators, the energy barrier corresponds to that of the transpose codes of these classical codes, paralleling the behavior in standard hypergraph products. In the 4D case, we showed that X and Z logical operators have energy barriers lower-bounded by $\Omega(d_\ell)$ and $\Omega(d_\ell^T)$, respectively.

Several questions remain for further investigation. An immediate one concerns the mechanism behind energy barrier enhancement in HHGP codes. A possible explanation lies in the increased check redundancy introduced by these higher-dimensional structures. This raises follow-up questions: Is there a quantitative relationship between redundancy and energy barriers? What amount or structure of redundancy is sufficient to achieve a specified energy barrier?

The redundancy provided by higher-dimensional hypergraph products is sufficient for achieving soundness and confinement~\cite{campbell2019theory,Quintavalle2021single}. Several works have also demonstrated the effectiveness of redundant parity check sets for simultaneously handling measurement and qubit errors~\cite{fujiwara2014ability,Ashikhmin2020quantum,Delfosse2022beyond}. However, establishing a general framework for introducing redundancy that ensures both confinement and macroscopic energy barriers remains an open question.

In contrast, quantum expander codes exhibit confinement even without check redundancy. This inspires the thought to explore alternative methods of providing confinement beyond redundancy or expansion, such as symmetries~\cite{sam2020symmetry,charles2023self}. Furthermore, quantum expander codes display unique statistical properties that may intrigue researchers, as demonstrated in recent studies~\cite{hong2024quantum,yin2024low,placke2024topological,deroeck2024ldpc}.  

Finally, quantum codes with macroscopic energy barriers are conjectured to support local decoders~\cite{Leverrier_2015,Fawzi_2018,kubica2019cellular,niko2017local,Hastings2014decoding,dennis2002topological,herold2015cellular,Breuckmann2018Scalable,dauphinais2017fault,Pastawski2011quantum,vasmer2021cellular}, based on intuition from statistical mechanics. In such systems, local cooling processes can reduce system energy, suggesting that for codes with macroscopic energy barriers, a simple iterative local process could function effectively as a decoder. However, the design of local decoders for higher-dimensional hypergraph product codes remains an open research question.

\section{Acknowledgments}
I thank Qi Ye of Tsinghua University for providing the argument for the energy barrier of two-dimensional repetition code. I thank Isaac Kim for helpful discussions and feedback on an earlier version of this manuscript. I also thank Andrew Doherty, Xiaotong Ni for valuable discussions. The author acknowledges the financial support from Sydney Quantum Academy. This work was supported by the Australian Research Council Centre of Excellence for Engineered Quantum Systems (EQUS, CE170100009), and partially supported by NSFC (No.92476206).

\section{Appendix}

Throughout this appendix, we work in the field $\mathbb{F}_2$. Thus, all addition operations are modulo $2$ except for the computation of the weight of vectors or matrices, i.e., the function $\mathrm{wt}(\cdot)$.

\subsection{Vector reshaping for 3D matrices}

Consider a basis $\mathcal{B}$ of the vector space $\mathbb{F}_2^{n_1} \otimes \mathbb{F}_2^{n_2}$:
\begin{eqnarray}
\mathcal{B}=\left\{a_i \otimes b_j \mid i= \{1, \ldots, n_1 \} \text { and } j= \{1, \ldots, n_2 \}\right\}.
\end{eqnarray}
Then any vector $v \in \mathbb{F}_2^{n_1} \otimes \mathbb{F}_2^{n_2}$ can be written as
\begin{eqnarray}
    v=\sum_{a_i \otimes b_j \in \mathcal{B}} v_{i j}\left(a_i \otimes b_j\right),
\end{eqnarray}
for some $v_{i j} \in \mathbb{F}_2$. We call the $n_1 \times n_2$ matrix $V$ with entries $v_{i j}$ the \emph{reshaping} of the vector $v$. By this definition, if $A, B$ are respectively $m_1 \times n_1$ and $m_2 \times n_2$ matrices, then 
\begin{eqnarray}
    (A \otimes B)v \Rightarrow A V B^T.
\end{eqnarray}

Define $\mathrm{wt}(M)$ as the number of ones in the vector (or matrix) $M$. We have $\mathrm{wt}((A \otimes B)v) = \mathrm{wt}(A V B^T)$.

Now consider a vector $v$ lies in the tensor product space $\mathbb{F}_2^{n_1} \otimes \mathbb{F}_2^{n_2} \otimes \mathbb{F}_2^{n_3}$, we need to extend the reshaping and transformation concepts to accommodate this higher-order tensor structure.

The basis for the tensor product space $\mathbb{F}_2^{n_1} \otimes \mathbb{F}_2^{n_2} \otimes \mathbb{F}_2^{n_3}$ is:
\begin{eqnarray}
    \mathcal{B}=\left\{a_i \otimes b_j \otimes c_k\right\},
\end{eqnarray}
for $i=\{1, \ldots, n_1 \}; j= \{1, \ldots, n_2 \}; k=\{1, \ldots, n_3\}$.

Thus, any vector $v$ in this space can be written as
\begin{eqnarray}
    v=\sum_{i=1}^{n_1} \sum_{j=1}^{n_2} \sum_{k=1}^{n_3} v_{i j k}\left(a_i \otimes b_j \otimes c_k\right),
\end{eqnarray}
where $v_{i j k} \in \mathbb{F}_2$. Here, the coefficients $v_{i j k}$ form a 3D matrix $V$ with shape $n_1 \times n_2 \times n_3$.

To analyze the action of a Kronecker product transformation, say $(A \otimes B \otimes C)$, we extend the idea of reshaping $v$: The vector $v$ can be interpreted as a 3D matrix $V$ with entries $v_{i j k}$. Let $A, B$, and $C$ be matrices of dimensions $m_1 \times n_1, m_2 \times n_2$, and $m_3 \times n_3$, respectively. The action of $A \otimes B \otimes C$ on $v$ can be expressed as
\begin{eqnarray}
    (A \otimes B \otimes C) v \equiv A V B^T C^T.
\end{eqnarray}
Here, $V$ is treated as a 3D matrix, and the transformation applies $A, B$, and $C$ along the first, second, and third modes of $V$, respectively. The resulting transformed 3D matrix $A V B^T C^T$ is a new 3D matrix of shape $m_1 \times m_2 \times m_3$.

The weight of the 3D matrix $V$, denoted wt $(V)$, counts the total number of ones in $V$. The Kronecker product transformation preserves the weight
\begin{eqnarray}
    \mathrm{wt}((A \otimes B \otimes C) v)=\mathrm{wt}\left(A V B^T C^T\right).
\end{eqnarray}

\subsection{Proof of Lemma~\ref{lemma:path_within_subset_3d}}
\label{sec:proofoflemma14}

Recall that a Z-type Pauli error of the 3D hypergraph product code can be expressed as a bit-string $z=\left(z^{(1)}, z^{(2)}, z^{(3)}\right)^T$. The corresponding energy $\epsilon(z)= \text{wt}(H_Xz)$ is
\begin{eqnarray}\label{eq:energy_expression_hgp_3d}
    \epsilon(z) &=& \text{wt}\left((I_{r_a}\otimes \delta_b \otimes I_{n_c})z^{(1)} + (\delta_a \otimes I_{r_b}\otimes I_{n_c}) z^{(2)}\right) \nonumber \\
    && + \text{wt}\left((I_{r_a}\otimes I_{n_b} \otimes \delta_c)z^{(1)} + (\delta_a \otimes I_{n_b}\otimes I_{r_c}) z^{(3)}\right) \nonumber \\
    && + \text{wt}\left((I_{n_a}\otimes I_{r_b} \otimes \delta_c)z^{(2)} + (I_{n_a} \otimes \delta_b\otimes I_{r_c}) z^{(3)}\right). \nonumber \\
\end{eqnarray}
By applying vector reshaping, the energy can be written as
\begin{eqnarray}\label{eq:energy_expression_hgp_reshape_3d}
    \epsilon(z) &=& \text{wt}\left(Z^{(1)}\delta_b^T + \delta_aZ^{(2)} \right) + \text{wt}\left(Z^{(1)}\delta_c^T + \delta_aZ^{(3)} \right) \nonumber \\
    &&+ \text{wt}\left(Z^{(2)}\delta_c^T + Z^{(3)}\delta_b^T \right),
\end{eqnarray}
where $Z^{(1)}$, $Z^{(2)}$ and $Z^{(3)}$ are the 3D matrices reshaped from $z^{(1)}$, $z^{(2)}$, and $z^{(3)}$, respectively. The $j$th a-slice of $Z^{(1)}$ is supported on qubit subset $r_1^j \otimes V_2 \otimes V_3$.

Using Eq.~\eqref{eq:energy_expression_hgp_reshape_3d}, we aim to prove a lower bound on the energy $\epsilon(z)$ [Lemma~\ref{lemma:pauli_weight_reduction_3d}]. To that end, we shall use the following convention. Let $L_a$ be a nontrivial codeword of $\delta_a^T$. 

The set of a-slices of $Z^{(1)}$ associated with the nonzero entries of $L_a$ will play an important role. We define the index set of such columns as $\mathcal{C}(L_a)$:
\begin{equation}
    \mathcal{C}(L_a):= \{j: (L_a)_j=1 \}.
\end{equation}

Given a codeword $L_a$ of $\delta_a^T$, one can construct a matrix $Z^{1, s}$ from the 3D matrix $Z^{(1)}$ by summing all the a-slices in the set $\mathcal{C}(L_a)$. More formally,
\begin{equation}
    Z^{1,s}_{jk} = \sum_{ i : i \in \mathcal{C}(L_a) } Z^{(1)}_{ijk},
    \label{eq:prescription_weight_reduction_3d}
\end{equation}
where the addition is modulo $2$. For example, if $L_a=110100$, we would sum a-slices $1, 2$, and $4$. Using Eq.~\eqref{eq:prescription_weight_reduction_3d}, we can deform an arbitrary path to a path consisting of Paulis only supported on subset $r_1^i \otimes V_2\otimes V_3$ for some $i$.~\footnote{The precise choice of $i$ does not matter; any $i\in \mathcal{C}(L_a)$ would suffice.} In particular, we can prove an inequality between the energy of the original Pauli and the deformed Pauli, proved in Lemma~\ref{lemma:pauli_weight_reduction_3d}. 
\begin{lemma}
\label{lemma:pauli_weight_reduction_3d}
    \begin{equation}
        \mathrm{wt}(Z^{1,s}\delta_b^T) \leqslant \mathrm{wt}\left(Z^{(1)}\delta_b^T + \delta_aZ^{(2)}\right).
    \end{equation}
\end{lemma}

\begin{proof}
    We prove this by contradiction. Consider the elements of $Z^{1,s}\delta_b^T$ and $Z^{(1)}\delta_b^T + \delta_aZ^{(2)}$. If $\mathrm{wt}(Z^{1,s}\delta_b^T) > \mathrm{wt}\left(Z^{(1)}\delta_b^T + \delta_aZ^{(2)}\right)$, then there exists $j,k$ such that
    \begin{eqnarray}
    \label{eq:weight_checks_rowspace_3d}
        \mathrm{wt}((Z^{1,s}\delta_b^T)_{j,k}) > \mathrm{wt}\left( (Z^{(1)}\delta_b^T + \delta_aZ^{(2)})_{j,k}\right).
    \end{eqnarray}
    We will prove that Eq.~\eqref{eq:weight_checks_rowspace_3d} cannot be satisfied, thereby proving the claim.

    Without loss of generality, consider the $j,k$ elements. Notice that $(Z^{1,s}\delta_b^T)_{jk}$ is a bit, thus the weight of it is either 0 or 1. If $\mathrm{wt}((Z^{1,s}\delta_b^T)_{jk})= 0$, Eq.~(\ref{eq:weight_checks_rowspace_3d}) cannot be satisfied. Therefore, we consider the $\mathrm{wt}((Z^{1,s}\delta_b^T)_{jk})= 1$ case.
    
    If $\mathrm{wt}((Z^{1,s}\delta_b^T)_{jk}) = 1$,  $(Z^{(1)}\delta_b^T)_{jk}$ must contain an odd number of ones on the a-slice in the set $\mathcal{C}(L_a)$. Otherwise, we would have had $\mathrm{wt}((Z^{1,s}\delta_b^T)_{jk}) = 0$, which is a contradiction. On the other hand, we remark that $(\delta_aZ^{(2)})_{jk}$ consists of an even number of ones on the a-slice set $\mathcal{C}(L_a)$. To see why, note that each $jk$-line of $\delta_aZ^{(2)}$ is a linear combination of the checks in $\delta_a^T$. Because $L_a$ is a codeword of $\delta_a^T$, $(\delta_aZ^{(2)})_{jk} L_a=0$. Therefore, in the $\mathrm{wt}((Z^{1,s}\delta_b^T)_{jk}) = 1$ case, the number of ones in $Z^{(1)}\delta_b^T + \delta_aZ^{(2)}$ on the $jk$-line and the a-slice in $\mathcal{C}(L_a)$ is odd. 

    Thus we conclude $\mathrm{wt}\left( (Z^{(1)}\delta_b^T + \delta_aZ^{(2)})_{jk}\right)\geq 1$. As such, Eq.~(\ref{eq:weight_checks_rowspace_3d}) cannot be satisfied. This completes the proof.
\end{proof}

With the same logic, one can prove that
\begin{lemma}
\label{lemma:pauli_weight_reduction_3d_2}
    \begin{equation}
        \mathrm{wt}(Z^{1,s}\delta_c^T) \leqslant \mathrm{wt}\left(Z^{(1)}\delta_c^T + \delta_aZ^{(3)}\right).
    \end{equation}
\end{lemma}

Now we are in a position to prove Lemma~\ref{lemma:path_within_subset_3d}. We do so by identifying a path $r'=\left\{P_0', P_1', \cdots, P_F'\right\}$ that is only supported on $r_1^\alpha \otimes V_2\otimes V_3$ while ensuring that $\epsilon_{\max }\left(r^{\prime}\right) \leqslant \epsilon_{\max }(r)$, for any path $r$ supported on the entire qubit set. Lemma~\ref{lemma:pauli_weight_reduction_3d} and Lemma~\ref{lemma:pauli_weight_reduction_3d_2} suggest a way to deform the path $r$ to the one supported on $r_1^\alpha \otimes V_2\otimes V_3$. 

Without loss of generality, let $r=\left\{P_0, P_1, \cdots, P_F\right\}$ be a path that $\Delta(L) = \epsilon_{\max}(r)$, with $P_0 = I$ and $P_F = L$. We consider a Pauli $P_i$ in the path $r$. It will be convenient to work in its binary representation, written as $(p^{(1)}, p^{(2)}, p^{(3)})^T$, where $p^{(1)}$,$p^{(2)}$ and $p^{(3)}$ represent the Paulis supported on $C_1\otimes V_2\otimes V_3$, $V_1\otimes C_2 \otimes V_3$ and $V_1 \times V_2 \otimes C_3$, respectively. First, we remove all the Paulis supported on $V_1\otimes C_2 \otimes V_3$ and $V_1 \times V_2 \otimes C_3$ by setting $p^{(2)}$ and $p^{(3)}$ as the zero vector. Next, we apply the following transformations to $p^{(1)}$. We reshape $p^{(1)}$ and refer to the reshaped 3D matrix as $P^{(1)}$. Let $L_a$ be a codeword of $\delta_a^T$ such that $(L_a)_{\alpha}=1$. Consider a set of a-slices in $P^{(1)}$ corresponding to the index set $\mathcal{C}(L_a)$. Denoting each a-slice as $r_1^k$, where $k \in \mathcal{C}(L_a)$, we update the a-slice $r_1^{\alpha}$ in the following way:
\begin{equation}
    r_1^{\alpha} \to r_1^{\alpha} + \sum_{k\in \mathcal{C}(L_a) \setminus \{ \alpha\}} r_1^k.
    \label{eq:column_update_3d}
\end{equation}
Afterward, the other a-slices of $P^{(1)}$ are set to the zero matrices. This yields the deformed Pauli operator $P_i'$.

By construction, the resulting $P_i^{\prime}$ is supported on $r_1^\alpha \otimes V_2 \otimes V_3$. Note that $r'$ is a valid path because $\mathrm{wt}(P_i^{\prime} P_{i+1}^{\prime}) \leq 1$ for every $i$. Also, because $P_F' = P_F = L$, $r'$ is still a path for $L$. Moreover, because $\epsilon\left(P_i^{\prime}\right) \leqslant \epsilon\left(P_i\right)$ for all $i$ [Lemma~\ref{lemma:pauli_weight_reduction_3d} and ~\ref{lemma:pauli_weight_reduction_3d_2}] as
\begin{eqnarray}
    \mathrm{wt}(Z^{1,s}\delta_c^T) + \mathrm{wt}(Z^{1,s}\delta_b^T)
    \leqslant \epsilon(z) = \mathrm{wt}\left(Z^{(1)}\delta_b^T + \delta_aZ^{(2)}\right) \nonumber \\
    + \mathrm{wt}\left(Z^{(1)}\delta_c^T + \delta_aZ^{(3)}\right) + \text{wt}\left(Z^{(2)}\delta_c^T + Z^{(3)}\delta_b^T \right), \nonumber \\
\end{eqnarray}
we have $\epsilon_{\max }\left(r^{\prime}\right) \leqslant \epsilon_{\max }(r)$. Both $r'$ and $r$ are paths for $L$, by definition  $\epsilon_{\max }\left(r^{\prime}\right) \geqslant \epsilon_{\max }(r)$, we conclude $\epsilon_{\max }\left(r^{\prime}\right) = \epsilon_{\max }(r)$.  Thus, by deforming $r$, we obtained a new path $r'$ supported on $r_1^\alpha \otimes V_2 \otimes V_3$ that yields the energy barrier $\Delta(L)$.

This argument can be applied to prove similar lower bounds for logical operators on $V_1 \otimes r_2^\beta \otimes V_3$ and $V_1 \otimes V_2 \otimes r_3^\gamma$. To conclude, for any elementary canonical logical operator $L$ supported on $r_1^\alpha \otimes V_2 \otimes V_3$ (resp. $V_1 \otimes r_2^\beta \otimes V_3$ and $V_1 \otimes V_2 \otimes r_3^\gamma$), their energy barrier can be given by a path supported on $r_1^\alpha \otimes V_2 \otimes V_3$ (resp. $V_1 \otimes r_2^\beta \otimes V_3$ and $V_1 \otimes V_2 \otimes r_3^\gamma$).

\subsection{Proof of Lemma~\ref{lemma:energy_barrier_of_composite_logical_3d}}
\label{sec:proofoflemma15}

Any nontrivial canonical logical-Z operator $L$ belongs to one of the following categories:
\begin{itemize}
    \item Case 1: $L$ is supported solely on the qubit subset $C_1 \otimes V_2 \otimes V_3$.
    \item Case 2: $L$ is supported solely on the qubit subset $V_1 \otimes C_2 \otimes V_3$.
    \item Case 3: $L$ is supported solely on the qubit subset $V_1 \otimes V_2 \otimes C_3$.
    \item Case 4: $L$ is supported on multiple subsets.
\end{itemize}
We will focus solely on Case 1. Case 2 and Case 3 can be analyzed similarly by considering subsets $V_1 \otimes C_2 \otimes V_3$ and $V_1 \otimes V_2 \otimes C_3$, while case 4 can be treated as Case 1, 2 or 3.

Without loss of generality, let the energy barrier of $L$ be attained by a path $r = \{P_0, P_1, \cdots, P_F\}$, with $P_0 = I$ and $P_F = L$. Similar to the approach taken in Lemma~\ref{lemma:path_within_subset_3d}, we aim to deform the path $r$ to the one supported on $r_1^k \otimes V_2 \otimes V_3$ for some $k$, such that the energy barrier of the deformed path lower bounds that of the $r$.

The deformation works in the same way as in the proof of Lemma~\ref{lemma:path_within_subset_3d}. We describe this procedure again for the readers' convenience.  Let $L_a$ be a nontrivial codeword of $\delta_a^T$ and $\mathcal{C}(L_a)$ be its corresponding column index set. We consider a binary representation of a Pauli $P_i$, written as $(p^{(1)}, p^{(2)}, p^{(3)})^T$. As in the proof of Lemma~\ref{lemma:path_within_subset_3d}, we remove the Paulis supported on $V_1 \otimes C_2 \otimes V_3$ and $V_1 \otimes V_2 \otimes C_3$ by setting $p^{(2)}, p^{(3)}$ as the zero vector. Next, reshape $p^{(1)}$ into a 3D matrix $P^{(1)}$ and update its columns in the following way. Choose $\alpha \in \mathcal{C}(L_a)$. This column is updated as Eq.~\eqref{eq:column_update_3d}. The other columns of $P^{(1)}$ are converted to zero vectors.

Thanks to Lemma~\ref{lemma:pauli_weight_reduction_3d} and ~\ref{lemma:pauli_weight_reduction_3d_2}, we obtain a new path $r'=\{P_0', P_1', \cdots, P_F'\}$ supported on $r_1^\alpha \otimes V_2 \otimes V_3$ with the property $\epsilon_{\max}\left(r'\right) \leqslant \epsilon_{\max}\left(r\right)$. Note that $P_F'$, in the binary representation, is of the form $r_1^\alpha \otimes \bar{y}\otimes \bar{z}$, where $\bar{y}, \bar{z}$ are codewords of $\delta_b$ and $\delta_c$ respectively. Therefore, $P_F'$ is either a nontrivial elementary canonical logical operator or the identity. In the latter case, $P_F'$ is the trivial codeword (zero vector) in the binary representation. Henceforth, we denote this as $L' = P_F'$.

If $L'$ is nontrivial, we can use the relation between the energy barriers of $L$ and $L'$:
\begin{eqnarray}
    \Delta(L) = \epsilon_{\max}(r) \geqslant \epsilon_{\max}(r') \geqslant \Delta(L').
\end{eqnarray}
Because $L'$ is an elementary logical operator, $\Delta(L)$ is greater or equal to the minimum energy barrier of elementary canonical logical operators. Thus, if $L'$ is nontrivial, the proof follows immediately.

If $L'$ is an identity, the above argument does not work. Fortunately, it turns out that for any $L'$, one can choose $L_a$ (the codeword of $\delta_a^T$ used in the current proof) such that $L'$ is not an identity.

Without loss of generality, consider a canonical logical-Z operator $L$, expressed as 
\begin{equation}
    L = \sum_{i,j,k} \alpha_{ijk}(x_i\otimes \bar{y}_j \otimes \bar{z}_k, 0, 0)^T,
\end{equation}
where (i) $\delta_b \bar{y}_j = \delta_c \bar{z}_k =0$ and (ii) $x_i \notin \operatorname{Im}\left(\delta_a\right)$ are unit vectors. If a given path ends with $L$, its deformation (using Eq.~\eqref{eq:column_update_3d}) yields the following operator $L'$:
\begin{equation}
    L' = \sum_{j,k} c_{jk}(x_\alpha \otimes \bar{y}_j \otimes \bar{z}_k, 0, 0)^T,
\end{equation}
where $\alpha\in \mathcal{C}(L_a)$ and  $c_{jk}$ is defined as 
\begin{equation}
    c_{jk} := \sum_{i\in \mathcal{C}(L_a)} \alpha_{ikj}.
\end{equation}
Note that $L'$ is trivial if and only if $c_{jk}=0$ for all $j,k$. Therefore, we aim to prove that there exists a choice of $L_a$ such that $c_{jk}=1$ for at least one choice of $j,k$.

Let us prove the contrapositive. Suppose $c_{jk}=0$ for all $k$, for any choice of $L_a$. Consider the following vector:
\begin{equation}
    u_{jk} := \sum_{i} \alpha_{ijk} x_i.
\end{equation}
Note that $c_{jk} = u_{jk}^T L_a$. By our assumption, $c_{jk}=0$ for any choice of $L_a$, so the inner product of $u_{jk}$ with any codeword of $\delta_a^T$ must be zero. On the other hand, $u_{jk}$, if it is nonzero, must lie outside of $\text{Im}(\delta_a)$ by the definition of the $x_i$. Thus, $u_{jk}$ is not an element of the row space of $\delta_a^T$. However, this is a contradiction for the following reasons. For a linear code, let $H$ and $G$ be the parity check matrix and the generator matrix. Then $v^TG=0$ if and only if $v$ is a vector in the row space of $H$. In our setup, if $c_{jk} = 0$ for any $L_a$, then $u_{jk}^T G=0$.  This implies that $u_{jk}$ must be in the row space of $\delta_a^T$, which contradicts the fact that it lies outside of $\text{Im}(\delta_a)$. To conclude, there must be at least one choice of $k, j$ such that $c_{jk}=1$. Thus, there is always a choice of $L_a$ such that $L'$ is not an identity, thereby proving the claim.

Case 2 and Case 3 can be analyzed similarly to Case 1 by considering subset $V_1 \otimes C_2 \otimes V_3$ and $V_1 \otimes V_2 \otimes C_3$. For Case 4, one can treat it just as Case 1, Case 2, or Case 3. For example, when treating it as Case 1, the logical operator $L$ has a nontrivial part in the subset $C_1 \otimes V_2 \otimes V_3$. One can prove there exists a codeword $L_a$ of $\delta_a^T$, such that after the deformation, the resulting $L'$ is a nontrivial elementary canonical logical operator. 

\subsection{Discussion on Conjecture~\ref{mainconjt}}

Let $L_a$ and $L_b$ denote the logical operators of $\delta_a$ and $\delta_b$, with distances $d_a$, $d_b$ and energy barriers $E_a$, $E_b$, respectively. Our goal is to determine the energy barrier of the operator $L_c = L_a \otimes L_b$, which serves as a logical operator for $\delta_c$.

To compute the energy barrier of $L_c$, we consider a path $r = \{P_0, P_1, \cdots, P_F\}$ with $P_0 = I$ and $P_F=L_c$. Then
\begin{equation}
    \Delta(L_c) = \min \{ \epsilon_{\max}(r): r \in w(0, L_c)\}.
\end{equation}
Note that $r$ is a path that walks on the bit of code $\delta_c$. Given a $P_\ell \in r$, the energy (the number of violated checks) of $P_\ell$ is
\begin{eqnarray}
    \epsilon(P_\ell) &=& \wt(\delta_cP_\ell) \nonumber \\
    &=& \wt \left(\left(\begin{array}{c}
         \delta_a \otimes I_{n_b} \\
         I_{n_a} \otimes \delta_b
    \end{array}\right) P_\ell \right) \nonumber \\
    &=& \wt((\delta_a \otimes I_{n_b}) P_\ell) + \wt((I_{n_a}\otimes\delta_b)P_\ell).
    \label{eq:energy_pell}
\end{eqnarray}
To proceed, one can express the identity matrices $I_{n_a}$ and $I_{n_b}$ as
\begin{equation}
    I_{n_a} = \sum_{j=1}^{n_a} w_jw_j^T, \quad I_{n_b} = \sum_{k=1}^{n_b} u_ku_k^T.
\end{equation}
Eq.~(\ref{eq:energy_pell}) can then be further processed as
\begin{equation}
\begin{aligned}
    \epsilon(P_\ell) &= \mathrm{wt}\left(\delta_a \otimes I_{n_b} P_\ell\right) + \mathrm{wt}\left( I_{n_a} \otimes \delta_b P_\ell\right) \\
    &= \mathrm{wt}\left(\delta_a \otimes \sum_{k=1}^{n_b} u_ku_k^T P_\ell\right) + \mathrm{wt}\left( \sum_{j=1}^{n_a} w_jw_j^T \otimes \delta_b P_\ell\right) \\
    &= \mathrm{wt}\left(\sum_{k=1}^{n_a} (\delta_a \otimes  u_k)( I_{n_a} \otimes u_k^T) P_\ell\right) \\
    & \ \ + \mathrm{wt}\left( \sum_{j=1}^{n_a} (w_j\otimes \delta_b)(w_j^T \otimes I_{n_b}) P_\ell\right) \\
    &= \mathrm{wt}\left(\sum_{k=1}^{n_b} (\delta_a \otimes  u_k)P_\ell^k\right) + \mathrm{wt}\left( \sum_{j=1}^{n_a} (w_j\otimes \delta_b)P_\ell^j\right) \\
    &= \sum_{k=1}^{n_b} \mathrm{wt}\left(\delta_aP_\ell^k\right) + \sum_{j=1}^{n_a}\mathrm{wt}\left(  \delta_bP_\ell^j\right),
\end{aligned}
\end{equation}
where
\begin{equation}
    P_\ell^k = ( I_{n_a} \otimes u_k^T) P_\ell, \quad P_\ell^j = (w_j^T \otimes I_{n_b}) P_\ell.
\end{equation}

To calculate the energy barrier of $L_c$, we minimize $\epsilon_{\max }\left(P_{\ell}\right)$ over all possible $r$. Here, $\epsilon\left(P_{\ell}\right)$ is the sum of two components: $\sum_{k=1}^{n_b} \mathrm{wt}\left(\delta_a P_{\ell}^k\right)$ and $\sum_{j=1}^{n_a} \mathrm{wt}\left(\delta_b P_{\ell}^j\right)$. Although these two components depend on each other, we first consider minimizing them independently.

We can first minimize $\sum_{k=1}^{n_b} \mathrm{wt}\left(\delta_a P_{\ell}^k\right)$, followed by minimizing $\sum_{j=1}^{n_a} \mathrm{wt}\left(\delta_b P_{\ell}^j\right)$. To minimize $\sum_{k=1}^{n_b} \mathrm{wt}\left(\delta_a P_{\ell}^k\right)$, we apply $L_a$ along the path $r_a$, which corresponds to the energy barrier of $L_a$ for each $k$ sequentially. The maximum energy contribution for this part is $E_a$. Subsequently, we minimize $\sum_{j=1}^{n_a} \mathrm{wt}\left(\delta_b P_{\ell}^j\right)$ by setting the order of $k$ according to the path $r_b$ that gives the energy barrier of $L_b$. The highest energy contribution for this part is $d_a E_b$. In particular, when the second part reaches $d_a E_b$, the first part has an energy of 0. Therefore, this path results in an energy barrier of $\epsilon_{\max }(r)=d_a E_b$.
Alternatively, we can reverse the process by first minimizing $\sum_{j=1}^{n_a} \mathrm{wt}\left(\delta_b P_{\ell}^j\right)$ and then $\sum_{k=1}^{n_b} \mathrm{wt}\left(\delta_a P_{\ell}^k\right)$. Following a similar analysis, this approach leads to an energy barrier of $\epsilon_{\max }\left(r^{\prime}\right)=d_b E_a$.

To complete the argument, one needs to prove that there is no path with an energy barrier lower than $\min \left[d_a E_b, d_b E_a\right]$. We developed a program to verify the conjecture~\footnote{Github repository: https://github.com/Guangqi-Phys/Energy-Barrier-of-tensor-product-codes.git}. This program generates random codes with their logical operators, computes their energy barriers and distances, constructs the tensor product code, and calculates its energy barrier.

Computing the energy barrier of a general code is hard. Our implementation uses a best-first search algorithm that may not always identify the absolute minimum. In our tests across many random codes, most results support the conjecture. In the rare cases where the energy barrier of the tensor product code differs from our product-based estimate, the discrepancy is minor and can be attributed to computational limitations.

While our numerical simulations strongly support the conjecture, one needs to develop alternative combinatorial methods for a rigorous proof. We believe that solving this problem will not only strengthen our current findings but also have broader applications for tensor product codes, particularly given their close relationship to locally testable codes.

\bibliographystyle{apsrev4-1-etal-title}
\bibliography{ref}

\end{document}